\documentclass[11pt]{article}
\usepackage[latin9]{inputenc}
\usepackage{float}
\usepackage{authblk}
\usepackage{amsfonts}
\usepackage{amsmath}
\usepackage{amssymb}
\usepackage{mathtools}
\usepackage{color}
\usepackage{amsthm}
\usepackage{ifpdf}
\usepackage{todonotes}
\usepackage{graphicx,caption}
\usepackage{bm}
\usepackage[section]{placeins}
\usepackage{pdflscape,afterpage}
\usepackage[pdftex,
	colorlinks=True,
	citecolor=red,
  	bookmarks=true,
	unicode=true,
	pdfborder={0 0 0},
	pdfstartview={FitV},
	pdfpagelayout={TwoPageRight},
	pdftitle={},
	linktoc=all,
	linkcolor=blue,
	plainpages=false,
	pdfpagelabels,
	pdffitwindow=true]{hyperref}

\usepackage[left=2.5cm, right=2.5cm, top=3cm]{geometry}
\def\eqlaw{\stackrel{\mathrm{law}}{=}}

\ifpdf
\else
\fi
\theoremstyle{plain}
\newtheorem{theorem}{Theorem}[section]
\newtheorem{corollary}[theorem]{Corollary}
\newtheorem{lemma}[theorem]{Lemma}

\theoremstyle{definition}
\newtheorem{assumption}[theorem]{Assumption}

\theoremstyle{remark}
\newtheorem{remark}[theorem]{Remark}
\newcommand{\R}{\mathbb{R}}

\newcommand{\be}{\beta}

\renewcommand{\bar}{\overline}
\renewcommand{\tilde}{\widetilde}
\renewcommand{\hat}{\widehat}

\newcommand{\dd}{\mathrm{d}}

\newcommand{\bE}{\mathbb{E}}
\newcommand{\bP}{\mathbb{P}}
\newcommand{\cN}{\mathcal{N}}
\newcommand{\cF}{\mathcal{F}}
\newcommand{\cO}{\mathcal{O}}
\newcommand{\1}{\mathbf{1}}
\newcommand{\Sbs}{\mathcal{S}_\mathrm{BS}}
\newcommand{\Sloc}{\mathcal{S}_\mathrm{loc}}

\numberwithin{equation}{section}

\def\bh{\bar{h}}
\def\eps{\varepsilon}

\def\s{\sigma}
\def\W{\mathrm W}

\def\V{\mathrm V}
\def\v{\mathrm v}

\def\WW{\mathbf W}
\def\VV{\mathbf V}

\def\h{\mathrm h}

\def\M{\mathcal M}

\def\maj{>}
\def\mino{<}
\def\be{\begin{equation}}
\def\ee{\end{equation}}

\def\sigDup{\sigma_\mathrm{loc}}
\def\sigBS{\sigma_\mathrm{BS}}

\begin{document}
\title{\textbf{Local volatility under rough volatility}}


\author[1,2]{
	F.\ Bourgey\thanks{florian.bourgey@polytechnique.edu}}
\author[1]{
	S.\ De Marco\thanks{stefano.de-marco@polytechnique.edu}}
\author[3]{
	P.\ K.\ Friz\thanks{friz@math.tu-berlin.de}}
\author[4]{
	P.\ Pigato \thanks{paolo.pigato@uniroma2.it}}
\affil[1]{Centre de Math\'ematiques Appliqu\'ees (CMAP), CNRS, Ecole Polytechnique, 
Institut Polytechnique de Paris, France
}
\affil[2]{Bloomberg L.P., Quantitative Research, London, UK}
\affil[3]{Technische Universit\"{a}t Berlin and Weierstra{\ss}-Institut, Berlin, Germany}
\affil[4]{Department of Economics and Finance, Universit\`a Roma Tor Vergata, Rome, Italy}

\maketitle

\begin{abstract}
\noindent Several asymptotic results for the implied volatility generated by a rough volatility model have been obtained in recent years
(notably in the small-maturity regime), providing a better understanding of the shapes of the volatility surface induced by rough volatility models,
supporting their calibration power to SP500 option data.
Rough volatility models also generate a local volatility surface, via the so-called Markovian projection of the stochastic volatility.
We complement the existing results on implied volatility by studying
the asymptotic behavior of the local volatility surface generated by a class of rough stochastic volatility models, encompassing the rough Bergomi model.
Notably, we observe that the celebrated ``$1/2$ skew rule" linking the short-term at-the-money skew of the implied volatility to the 
short-term at-the-money skew of the local volatility, a consequence of the celebrated ``harmonic mean formula''
of [Berestycki, Busca, and Florent, QF 2002], 
is replaced by a new rule: the ratio of the at-the-money implied and local volatility
skews tends to the constant $1/(H + 3/2)$ (as opposed to
the constant $1/2$), where $H$ is the regularity index of the underlying instantaneous volatility process.
\end{abstract}


\section{Introduction}\label{sec:intro}
In rough stochastic volatility models, volatility is driven by a fractional noise, in the rough regime of Hurst parameter $H$ less than 1/2. With no claim of completeness, we mention econometric evidence \cite{gatheral2018volatility,FTW19,BLP16},  market microstructure foundations \cite{el2018microstructural}, efficient numerical methods and simulations schemes \cite{bayer2020,BLP15,McCrickerdPakkanen2018,euch2016characteristic}, including deep learning algorithms \cite{bayer2019deep, goudenege}.

This work is concerned with pricing under rough volatility, a key feature of which, well-adapted to the steep volatility skews seen in Equity option markets, is the power-law behavior
of the short-dated {\em implied volatility} at-the-money (ATM) skew: 
$$
                \Sbs \sim \mathrm{(const)} t^{H-1/2} \,;
$$
references include \cite{alos2007short, fukasawa2011asymptotic, bayer2016pricing, bayer2017short, forde2017asymptotics, euch2018short, fukasawa2017short, friz2021short, GULISASHVILI20203648, fukasawa2020,BHP21}.
More specifically, we consider here - to the best of our knowledge for the first time - the Dupire {\em local volatility} \cite{dupire1996unified,dupire1994pricing} generated by rough volatility models, using a Gyongy-type projection \cite{GyongyMimicking, BrunickShreve} and study its short-time behavior in a large deviations regime. In the context of {\em implied volatility}, such a regime was pioneered in \cite{forde2017asymptotics}, with instantaneous stochastic volatility process given as $v(t,\omega) = \sigma( W^H_t (\omega))$, i.e., as some explicit function of a fractional Brownian motion (fBm). In particular, no rough or Volterra stochastic differential equations need to be solved. We shall work under the same structural assumption as \cite{forde2017asymptotics}, although under less restrictive growth assumptions, such as to include the popular rough Bergomi model \cite{bayer2016pricing} where $\sigma$ has exponential form.  Postponing detailed recalls and precise definitions, our main result (Theorem \ref{th:asymptotic:markovian:projection}) states that, 
$$
       \sigDup \bigl(t, y \, t^{1/2-H} \bigr) 
\to \sigma \bigl(\, \hat{h}_1^y \, \bigr)
\quad \mbox{ as } t \downarrow 0,
$$
where $\hat{h}^y$ is related to a minimization problem, similar to a geodesic in Riemannian geometry.
Our analytic understanding is sufficiently fine to exploit it on the one hand for numerical tests  (discussed in Section \ref{sec:num}) 
and on the other hand to derive further analytic results (formulated in Sections \ref{sec:modeling_framework} and \ref{sec:main_results}, 
with proofs left to Section \ref{sec:proofs}) 
including the blowup, when $H < 1/2$, of the {\em local volatility skew} in the short-dated limit,
$$
                \Sloc \sim \mathrm{(const)} t^{H-1/2},
$$
see Corollary \ref{corollary:Hplus32} below for a precise statement and information on the constant.
This finding is consistent with \cite{fukasawa2017short} where it is shown, amongst others, that in ``regular'' local-stochastic vol models, which amounts to a regularity assumption on $\sigDup$, the implied volatility skew does not explode. 
The regularity of $\sigDup$ is violated here in the sense that $\Sloc$ is infinite at $t=0$. This is also consistent with \cite{pig19,fps2020} where it is shown that a  ``singular'' $\sigDup$ can indeed produce exploding implied skews. 

A further interesting consequence, also part of Corollary \ref{corollary:Hplus32}, is then that the $1/2$-rule of thumb from practitioners \cite{Derman96}
(see also \cite{gatheral2006volatility} and \cite[Remark 3.4]{fukasawa2020} for different proofs)
actually {\em fails} and is replaced, again in the short-dated limit, by what we may call the $1/(H+3/2)$-rule,
\be \label{e:skew_rule_intro}
\frac{\Sbs}{\Sloc} 
\to
\frac{1}{H+3/2}. 
\ee
As a sanity check, for Hurst parameter $H=1/2$ we are in a diffusive regime and then indeed fall back to the $1/2$-rule.

{\bf Techniques and further discussion.}
Our analysis is based on a mixture of {\it large deviations} (see e.g.\ \cite{friz2015large} for a recent collection with many references),  {\em Malliavin calculus} \cite{bally:inria-00071868,Nualart:06,fournie2001applications}, and last not least ideas from {\em rough paths} and {\em regularity structures} techniques, following \cite{bayer2017regularity, friz2018precise1,friz2021short}; see also Section 14.6 in \cite{friz2020course}. In order to deal with $H<1/2$, we cannot rely on previously used methods in diffusion settings such as \cite{takanobu1993asymptotic, de2018local}. Local volatility in classical stochastic volatility models, including Heston, is discussed in many books on volatility modeling, \cite{gatheral2006volatility} remains a key reference. Rigorous asymptotic results include \cite{ghlow, de2013rational,de2018local}. In affine forward variance models, including rough Heston \cite{euch2016characteristic, gatheral2019affine}, it is conceivable that saddle-point-based techniques, in the spirit of \cite{de2013rational} could be employed to study local volatility asymptotics. The bottleneck in such an approach seems however the lack of explicit knowledge of the moment-generating function, only given implicitly via convolution Riccati equations. 
We note that the recent preprint \cite{AlosEtAl2022} confirmed the asymptotic result \eqref{e:skew_rule_intro} using some representations of  $\Sbs$ and $\Sloc$ based on Malliavin calculus,
in a central limit (Edgeworth) regime, as opposed to our large deviations setting.

\medskip

{\bf Acknowledgements:} PKF gratefully acknowledges financial support from European Research Council (ERC) Grant CoG-683164 and German science foundation (DFG) via the cluster of excellence MATH+, project AA4-2.
SDM gratefully acknowledges financial support from the research projects Chaire Risques Financiers (\'{E}cole Polytechnique, Fondation du Risque and Soci\'{e}t\'{e} G\'{e}n\'{e}rale) and Chaire Stress Test, Risk Management and Financial Steering (\'{E}cole Polytechnique, Fondation de l'\'{E}cole Polytechnique and BNP Paribas). PP gratefully acknowledges financial support from INDAM-GNAMPA.
We thank Martin Forde, Masaaki Fukasawa, Paul Gassiat, Antoine Jacquier, Fabio Mercurio, Olivier Prad\`ere, Sergio Pulido, Adil Reghai, Mathieu Rosenbaum and Guillaume Sebille for feedback and stimulating discussions on the subject of this article, and Andrea Pallavicini and Riccardo Longoni for their inspiring comments and insights.

\section{The modeling framework}\label{sec:modeling_framework}
\noindent We assume $S_0=1$ and that the log price $X_t := \log S_t$ satisfies
\begin{equation}\label{eq:rvol:model}
\begin{split}
\dd X_t & = -\frac12 V_t  \dd t + \sqrt{V_t} \left( \rho \, \dd W_t  + \bar{\rho} \, \dd \bar{W}_t \right)
\\
V_t &= \sigma^2(\hat{W}_t)
\end{split}
\end{equation}
with `volatility function' $\sigma: \R \to \R$. We shall assume $\sigma$ to be smooth, subject to mild growth conditions
given below, such as to cover rough Bergomi type situations where $\sigma (x) = \sigma_0 \exp (\eta x)$.


\begin{assumption}\label{assu:sigma}
	There exist $c_1,c_2,c_3,c_4 > 0$ such that for all $x\in \R$,
	\begin{align}
		c_1 e^{- c_2 |x|} &\leq \sigma( x ), \label{eq:C1}
		\\
		\sigma(x) &\leq c_3 e^{ c_4 |x|}. \label{eq:C2}
	\end{align}
\end{assumption}

We take $\rho^2+\bar{\rho}^2=1$, with $\rho \in(-1,1)$. 
We denote  $\W = (W, \bar W)$ where $W, \bar W$ are two independent 
standard Brownian motions.
These are used to construct
\begin{equation}\label{eq:RLfbm}
\tilde W_t = \bar{\rho} \bar W_t+\rho W_t \ \ \ \text{and} \ \ \    \hat W_t = (K * \dot W)_t=\int_0^t K(t,s) \dd W_s,
\end{equation}
with 
\[
K(t,s)=\sqrt{2H}(t-s)^{H-1/2} 
\qquad \mbox{for } t > s 
\]
and $K(t,s) = 0$ otherwise, so that $\tilde W$ is again a standard Brownian motion ($\rho$-correlated with $W$) whereas $\hat W$ 
is a Riemann--Liouville fBm with Hurst index $H \le 1/2$, i.e., the self-similar Gaussian Volterra process in \eqref{eq:RLfbm}.

We will use analogous notations for Cameron--Martin paths
$ \h = (h, \bar h)$, so that $\tilde h = \bar{\rho} \, \bar h+\rho \, h$, and $\hat h_t = (K^H * \dot h)_t = \int_0^t K(t,s) \dd h_s$. We denote $H^1$ the Cameron--Martin space and  $\| \cdot \|_{H^1}$ the Cameron--Martin norm $\| \h \|_{H^1}^2 = \int_0^1 (\dot h^2 +  \dot{\bar{h}}^{2}) \dd t$.

\section{Mathematical setting and results}\label{sec:main_results}
\noindent The time-scaling property of the Gaussian process	$(W, \bar W, \hat W)$ underlying the model \eqref{eq:rvol:model}
yields $X_{\eps^2} \eqlaw X_1^{\varepsilon}$ for every $\eps > 0$, where $X_1^{\varepsilon}$ satisfies
\begin{equation}\label{eq:rescX}
X_1^{\varepsilon} 
=
\int_{0}^1 \sigma \left( \eps^{2H} \hat W_s \right) \varepsilon \,
\dd \left(\bar{\rho} \, \bar W+\rho \, W\right)_s -\frac{1}{2}\varepsilon ^{2}\int_{0}^1
\sigma^2 \left(   \eps^{2H} \hat W_s\right) \dd s \,.
\end{equation} 
Forde and Zhang proved in \cite{forde2017asymptotics}, albeit under different technical conditions on the volatility function,
that a small noise Large Deviation Principle (LDP) holds for the 
family $\eps^{2 H - 1} X^{\eps}_1$ (hence for $\eps^{2 H - 1} X_{\eps^2}$) as $\eps \to 0$, with speed $\eps^{4H}$ and rate function 
\begin{align} \label{e:rateF}
\Lambda(y) & :=\inf_{\h = (h, \bar h) \in H^{1}}\left\{
\frac 12 \| \h \|_{H^1}^2 :
\int_0^t \sigma\left(\hat h_s\right) (\rho \,\dd h_s + \bar \rho \, \dd \bar h_s)
= y\right\}
= 
\frac 12  \| \h^y \|_{H^1}^2 \,,
\end{align}
where $\h^y$ is a minimizer of the control problem defining $\Lambda(y)$.
From the LDP \eqref{e:rateF}, we have
\begin{align}\label{eq:ldp:sn}
     - \eps^{4H} \log \bP \bigl( X_1^\eps \geq y \eps^{1-2H} \bigr) &\rightarrow  \Lambda(y)= \frac{1}{2} \| { \h}^y  \|^2_{H^1},  \mbox{ for } y \geq 0 \mbox{ as } \eps \downarrow 0,
\\
     - \eps^{4H} \log \bP \bigl( X_1^\eps \leq y\eps^{1-2H} \bigr) &\rightarrow \Lambda (y)=
     \frac{1}{2} \| { \h}^y  \|^2_{H^1}, \mbox{ for } y \leq 0, \mbox{ as } \eps \downarrow 0 \,,
\end{align}
and this small-noise LDP eventually translates to a short-time LDP for the process $X_{\eps^2}$.
This result was proved in the case where $V_t = \sigma(\hat{W}_t)$ in \cite{forde2017asymptotics}, and then extended to the possible time dependence of the form $V_t = \sigma(\hat{W}_t,t^{2H})$ in \cite[Section 7.3]{friz2018precise1} (see also Remark \ref{rm:timedependence} below).

The short-time result for call and put prices reads as follows (see \cite[Corollary 4.13]{forde2017asymptotics})
\begin{align}
\label{eq:ldp:call}
     - t^{2H} \log \bE \bigl[ (e^{X_t} - e^{y \, t^{1/2-H}})^+ \bigr] \rightarrow \Lambda(y)= \frac{1}{2} \| { \h}^y  \|^2_{H^1} , \mbox{ for } y > 0 \mbox{ as } t \downarrow 0,
\\
\label{eq:ldp:put}
     - t^{2H} \log \bE \bigl[ (e^{y \, t^{1/2-H}}-e^{X_t})^+ \bigr] \rightarrow \Lambda(y)= \frac{1}{2} \| { \h}^y  \|^2_{H^1}, \mbox{ for } y < 0 \mbox{ as } t \downarrow 0 \,,
\end{align}
where $\h^y$ is as above.
Let us also recall that these option price asymptotics imply the following asymptotic formula for the Black--Scholes implied volatility 
(notation: $\sigma_{BS}$), which can be seen as a ``rough'' version of the Berestycki--Buscat--Florent (BBF)
formula \cite{berestycki2004computing}:
\begin{equation}\label{eq:asy:bs}
\sigBS^2(t, yt^{1/2-H} ) \to \chi^2(y):=\frac{y^2}{2\Lambda(y)}\quad 
\mbox{ for } y \neq 0 
\mbox{ as } t\downarrow 0.
\end{equation}

\begin{remark}[Precise conditions for the LDP, call price asymptotics and implied volatility asymptotics]\label{rm:condLDP}
The exponential growth condition \eqref{eq:C2} is no obstruction for an LDP to hold for the model \eqref{eq:rvol:model}, as was shown in \cite{bayer2017regularity, gulisashvili2017large}, 
weakening the linear growth condition first required in \cite{forde2017asymptotics}.
Moreover, while the put price asymptotics \eqref{eq:ldp:put} always holds, the unboundedness of the call option payoff requires some additional 
condition for \eqref{eq:ldp:call} to hold: with reference to \cite[Assumption A2]{friz2018precise1}, we will assume the following ``$1+$ moment condition'' whenever necessary:
\begin{assumption}
	\label{assu:moment}
	There exists $p \maj 1$ such that $\limsup_{\eps \to 0} \bE[e^{p X^\eps_1}] \mino \infty$.
\end{assumption}
Following \cite[Lemma 4.7]{friz2018precise1},  Assumption \ref{assu:moment} is true under the following stronger, but more explicit, condition: the process $S_t = e^{X_t}$ is a martingale, and there exist $p \maj 1$ and $t \maj 0$ such that $\bE [S_t^p] \mino \infty$.
It is known that such a condition on the moments of $e^{X_t}$ is satisfied when $\sigma$ has linear growth, cf.\ \cite{forde2017asymptotics}, while in the case $H=1/2$, the same is true under much weaker assumptions ($\sigma$ of exponential growth and $\rho <0$ is enough, see \cite{sin1998complications, jourdain2004loss}). We expect similar results to hold for $H<1/2$, but they have not been proved yet; see the partial results available in \cite{gassiat2018martingale, GULISASHVILI20203648}. 
\end{remark}

The Markovian projection of the instantaneous variance $V_t$ 
(see \cite{GyongyMimicking},\cite[Corollary 3.7]{BrunickShreve}) 
within the model \eqref{eq:rvol:model} is defined by
\be 
\label{e:MarkovProj}
\sigDup^2(t, k) :=
\bE \bigl[
V_t |X_t = k
\bigr] 
\qquad \text{for every } t \maj 0 \text{ and } k \in \R.
\ee
It follows from references \cite{GyongyMimicking, BrunickShreve} that the dynamics of the resulting local volatility model are weakly well-posed; see also \cite{friz2014make} for a generic regularization scheme obtained by time-shifting the local volatility surface (a procedure that we do not require here). 

We now present our main result. We prove that the local volatility function \eqref{e:MarkovProj} satisfies the following  short-time asymptotics.

\begin{theorem}[Markovian projection at the LDP regime] \label{th:asymptotic:markovian:projection}
Let Assumption \ref{assu:sigma} be in force.
Then, the Markovian projection in the model  \eqref{eq:rvol:model} satisfies, for every $y\in\R\setminus\{0\}$ small enough,
\begin{equation}\label{eq:asy:loc}
\sigDup^2 \bigl(t, y \, t^{1/2-H} \bigr) =
\mathbb{E}\bigl[\bigl.V_{t}\bigr|X_{t}=yt^{1/2-H}\bigr]
\to \sigma^2 \Bigl(\hat{h}_1^y \Bigr)\quad \mbox{ as } t\downarrow 0 \,,
\end{equation}
where we recall that $\hat h^y_t = \int_0^t K(t,s) \dd h^y_s$ 
and 
$\h^y = (h^y, \bar h^y)$ is the minimizer of the rate function in \eqref{e:rateF}.
\end{theorem}

The uniqueness of the minimizer for the control problem \eqref{e:rateF} is proved in \cite[Lemma C.6]{friz2018precise1}.
Let us stress that the asymptotics \eqref{eq:asy:loc} for the local volatility function holds under 
the mild growth conditions
of Assumption \ref{assu:sigma}, while we do not require the $1+$ moment condition of Assumption \ref{assu:moment}. 

\subsection{Local volatility skew and the new $\frac{1}{H+3/2}$ rule} 
Let us write $\sim$ for asymptotic equivalence as $t\to 0$.  Let us denote
\[
\Sigma(y) := \sigma \Bigl(\hat{h}_1^y \Bigr)
\]
the limiting function in \eqref{eq:asy:loc}, and consider the following finite-difference approximations of the local and
implied volatility skew
\begin{align}\label{eq:locvolskew}
\Sloc(t, y) 
&:=
\frac{
	\sigDup(t, y \, t^{1/2 - H} ) - \sigDup(t, -yt^{1/2 - H})
} {2y \, t^{1/2 - H}},
\\
\Sbs(t, y)
&:=
\frac{
	\sigBS(t, y \, t^{1/2 - H} ) - \sigBS(t, -yt^{1/2 - H})
} {2y \, t^{1/2 - H}}. 
\end{align}
Then, we have the following

\begin{corollary}[Local vol skew and the new $\frac{1}{H+3/2}$ rule]\label{corollary:Hplus32}
Let $\rho\neq 0$. Let Assumption \ref{assu:sigma} be in force.
Then, for $y\in\R\setminus\{0\}$ small enough,
\begin{equation}\label{eq:ldp:loc:vol:skew}
\Sloc(t, y)
\sim 
\frac{\Sigma(y) - \Sigma(-y)}{2y} \frac 1 {t^{1/2 - H}}
\end{equation}
as $t \to 0$.
Under the additional moment condition in Assumption \ref{assu:moment},
\begin{equation}\label{eq:1_H32rule}
\frac{\Sbs(t, y)}
{\Sloc(t, y)}
\xrightarrow{t\to 0}
\frac{\chi(y) - \chi(-y)}
{\Sigma(y) - \Sigma(-y)}
\xrightarrow{y\to 0}
\frac{1}{H+3/2}.
\end{equation}
In the case $\rho = 0$, we have $\Sbs(t, y) = 0$ and $\Sloc(t, y) = 0$ for every $t$.
\end{corollary}
\noindent In our numerical experiments in section \ref{sec:num}, we estimate the exact ATM local volatility skew $\frac 12 \partial_k \sigDup(t,k)\big|_{k = 0}$ in the rough Bergomi model \eqref{e:rBergomi}, and find perfect agreement with Corollary \ref{corollary:Hplus32}. The model local volatility skew can be observed in Figure \ref{fig:ATMskews}, and the ratio of the implied volatility skew over the local volatility skew in Figure \ref{fig:skews_ratio}.

\begin{remark} \label{rem:one_half_rule}
When $H=1/2$, we are back to the classical $1/2$ skew rule, see 
Derman et al.\ \cite{Derman96}.
\end{remark}

\begin{remark}
One can expect the $\frac{1}{H+3/2}$ rule \eqref{eq:1_H32rule} to hold also for rough or rough-like volatility models that do not belong to the model class \eqref{eq:rvol:model}, such as the rough Heston model \cite{euch2016characteristic}.
The recent preprint \cite{DallAcqua2022} provides numerical evidence for the $\frac{1}{H+3/2}$ rule under the lifted Heston model \cite{AbiJaber2019}, a Markovian approximation of rough Heston, as well as a formal proof in the case of the proper rough Heston model, see \cite[Proposition 2.1]{DallAcqua2022}.
In their recent work \cite{AlosEtAl2022}, Alos and co-authors prove the $\frac{1}{H+3/2}$ rule for stochastic volatility models under suitable assumptions on the asymptotic behavior of the volatility process and related iterated Malliavin derivatives, further providing an asymptotic rule for 
 the ratio of the at-the-money second derivative $\partial_{kk} (\cdot) |_{k=0}$ of the local and implied volatility functions.
\end{remark}

\begin{remark}[The short-time harmonic mean formula and the 1/2 skew rule again] \label{rem:harm_mean}
When expressed in terms of an implied volatility $\sigBS$, Dupire's formula for local volatility reads 
\be \label{e:Dup}
\sigDup(t,k)^2 = 
\frac{\sigBS(t,k) + 2 \, t \, \partial_t \sigBS(t,k)}
{\Bigl(
t \, \partial_{kk} \sigBS
- \frac 14 t^2 \, \sigBS (\partial_k \sigBS)^2
+ \frac 1 {\sigBS} \left(1 - \frac{k \, \partial_k \sigBS}{\sigBS}\right)^2 \Bigr)(t,k)} 
\ee
provided that $\sigBS$ is sufficiently smooth for all the partial derivatives to make sense.
Formally taking $t\to 0$ inside \eqref{e:Dup} and assuming that the partial derivatives $\partial_t \sigBS$, $\partial_k \sigBS$ and $\partial_{kk} \sigBS$
remain bounded,
one obtains
\begin{equation}\label{eq:ode:sigloc:sigbs}
	\sigDup(0,k)^2 = \frac{\sigBS(0,k)^2}{\bigl(1 - k \frac{\sigBS^{\prime} (0,k) }{\sigBS(0,k)}\bigr)^2}.
\end{equation}
The ordinary differential equation \eqref{eq:ode:sigloc:sigbs}
can be used to reconstruct the function $\sigBS(0, \cdot)$ from $\sigDup(0,\cdot)$ and it is solved
by
the harmonic mean function
\begin{equation}\label{eq:def:harmonic:mean}
	H(t,k) = \frac 1{\frac 1k \int_0^k \frac {\dd y}{\sigDup(t,y)}},
\end{equation}
evaluated at $t=0$.
The computation above, leading from \eqref{eq:ode:sigloc:sigbs} to \eqref{eq:def:harmonic:mean}, can be found in \cite{Lee2005}; the rigorous counterpart of this formal argument, that is the asymptotic equivalence $\sigBS(t, k) \sim H(t,k)$, 
known as the ``harmonic mean formula'' or BBF formula, was proven in \cite{BBF02}
under the assumption that the local volatility surface $\sigDup$ is bounded and uniformly continuous in a neighborhood of $t=0$.
It is straightforward to see that the harmonic mean satisfies the property $\partial_k H(t,k)\big|_{k=0} = \frac 12 \partial_k \sigDup(t,k)\big|_{k = 0}$.
Therefore, if we assume that the short-time approximation property $\sigBS(t,k) \approx H(t,k)$ also holds for the first derivatives with respect to $k$, as a consequence we obtain the $1/2$ short-time skew rule $\partial_k \sigBS(t,0) \sim 
\frac 12 \partial_k \sigDup(t,0)$ that we referred to in Remark \ref{rem:one_half_rule}.

Corollary \ref{corollary:Hplus32} entails that the formal argument above does not hold anymore for the implied and local volatility surfaces generated by a rough stochastic volatility model. 
Notably, the boundedness of the partial derivatives $\partial_t \sigBS$, $\partial_k \sigBS$ and $\partial_{kk} \sigBS$, and the uniform continuity of the local volatility surface, fall short -- but in such a way that the limit of the skew ratio $\frac{\Sbs}{\Sloc}$ can still be identified and explicitly computed (see related numerical tests in Figure \ref{fig4}).
\end{remark}

\begin{remark}[Time-dependent volatility function] \label{rm:timedependence}
The rough Bergomi model \cite{bayer2016pricing} comes with instantaneous variance
\[
\xi (t)  \exp \Bigl(\eta x - \frac{\eta^2}2	t^{2 H} \Bigr) \sim \xi (0)  \exp \bigl(\eta x  \bigr) =: \sigma^{2}(x)
\]
as $t \downarrow 0$. We could have proved Theorem \ref{th:asymptotic:markovian:projection} and Corollary \ref{corollary:Hplus32} in greater generality, with $V_t = \sigma^2(\hat{W}_t, t)$, provided the dependence with respect to $t$ in $\sigma = \sigma(x,t)$ is sufficiently smooth such as not to affect the local analysis that underlies the proof. 
This is more subtle in the case of rough Bergomi where $t^{2H}$ fails to be smooth at $t=0^+$ when $H<1/2$.
Even so, we discussed in \cite{friz2018precise1} how to adjust the arguments to obtain exact asymptotics, the same logic applies here.
\end{remark}

\section{Numerical tests}\label{sec:num}

\noindent We wish to estimate the conditional expectation \eqref{e:MarkovProj} for some specific instance of the model \eqref{eq:rvol:model}, 
using Monte Carlo simulation.
We consider the rough Bergomi model \cite{bayer2016pricing}, for which the instantaneous variance process is given by
\be \label{e:rBergomi}
V_t = \xi_0 \exp \Bigl(\eta \hat W_t - \frac{\eta^2}2	t^{2 H} \Bigr)
=  \xi_0 \exp \Bigl(\eta \int_0^t \sqrt{2H} (t-s)^{H - 1/2} \dd W_s  - \frac{\eta^2}2	t^{2 H} \Bigr),
\ee
where $\xi_0 = V_0$ is the spot variance and $\eta$ a parameter that tunes the volatility of variance.
Note that, strictly speaking, Theorem \ref{th:asymptotic:markovian:projection} and Corollary \ref{corollary:Hplus32} do not apply to the model above,
because of the time dependence in the volatility function $\sigma(x,t) = \sqrt{\xi_0} \exp \bigl(\frac{\eta}{2} x - \frac{\eta^2}4	t^{2 H} \bigr)$.
In light of the discussion in Remark \ref{rm:timedependence}, we can expect our asymptotic results to hold for such a time-dependent volatility function as well, which is
in line with the output of our numerical experiments below.

For a given time horizon $T>0$ and a number $N \in \mathbb{N}^{*}$ of time-steps, the random vector $(\log V_{t_k})_{1 \le k \le N}$, $t_k = k \frac T N$, has a multivariate Gaussian distribution with known mean and variance, see for example \cite{bayer2016pricing}, and can therefore be simulated exactly. We use the standard simulation method for Gaussian vectors based on a Cholesky factorization of the covariance matrix. Of course, this method has a considerable complexity  -- cost $\cO(N^3)$ for the Cholesky factorization  and $\cO(N^2)$ for the matrix multiplication required to get one sample of  $(V_{t_k})_{0 \leq k \leq N}$ --  but our focus is on the accuracy of our estimations, rather than on their computational time.
We construct approximate samples of  the log-asset price
$X_{T} = -\frac{1}{2}\int_{0}^{T}V_{t}\dd t + \int_{0}^{T}\sqrt{V_{t}}(\rho \dd W_{t} + \bar \rho \dd \bar{W}_{t})$
using a forward Euler scheme on the same time-grid
\[
X^N_{T} = - \frac{T}{2 N} \sum_{k=0}^{N-1} V_{t_k} 
+ \sum_{k=0}^{N-1} \sqrt{V_{t_k}} \Bigl(\rho (W_{t_{k+1}} - W_{t_k}) + \bar \rho (\overline W_{t_{k+1}} - \overline W_{t_k}) \Bigr).
\]
Therefore, we obtain $M$ i.i.d.\ approximate Monte Carlo samples $(X^{N,m}_{T}, V^m_T)_{1 \le m \le M}$ of the couple $(X_T^N, V_T)$, from which
our estimators of the implied volatility  and local volatility \eqref{e:MarkovProj} are constructed, as detailed below.
Since our goal is to check the asymptotic statements appearing in Theorem \ref{th:asymptotic:markovian:projection} and 
Corollary \ref{corollary:Hplus32},  we will consider a large number $N$ of discretization steps and a large number $M$
of Monte Carlo samples in order to increase the precision of the estimates we use as a benchmark.
We estimate out-of-the-money put and call option prices by standard empirical means and evaluate the corresponding implied volatilities
$\sigBS(T,K)$ by Newton's search.

The rough Bergomi model \eqref{e:rBergomi} parameters we used in our experiments are $S_0 = 1, \eta = 1.0, \rho = -0.7$, and $\xi_0 = 0.235^2$.
We tested three different values of $H \in (0,1/2]$, namely $H \in \{0.1, 0.3, 0.5\}$. We used $M = 1.5 \times 10^6$ Monte Carlo samples and $N = 500$ discretization points.

\begin{remark} 
\label{rem:weak_error_rate}
Several recent works \cite{bayer2020weak, bayer2022weak, gassiat2022weak, friz2022weak} study the weak error rate of rough Bergomi type models. Without going into (bibliographical)
details, the weak rate has now been identified as $1$ for $H$ above $\frac 1 6$ and  $3H+\frac{1}{2}$ for $H$ below $\frac 1 6$. Importantly, as $H \downarrow 0$, a weak rate of $\frac 12$ persists. 
The fairly large number of time steps we considered in our experiments ($N = 500$) is arguably enough to obtain good benchmark values when $H$ is close to $\frac 1 2$, but we should bear in mind that the bias in the Monte Carlo estimation is expected to become more and more important as $H$ approaches zero. In this case, larger number of time steps might be required to get a trustworthy level of accuracy; of course, the complexity of the exact Cholesky method we exploited in our simulation of the Riemann--Liouville process makes the simulations very demanding for very large values of $N$.
\end{remark} 


\subsection{Local and implied volatility estimators}

In this section, we present in detail the estimators we have implemented for the target objects:  the at-the-money implied volatility skew $\partial_{k}\sigBS(t,k)|_{k=0}$, the local volatility function (or Markovian projection) $\sigDup(\cdot, \cdot)$ in \eqref{e:MarkovProj}, and the local volatility skew $\partial_{k}\sigDup(t,k)|_{k=0}$.

\paragraph{The estimator of the implied volatility skew.}
A representation of the first derivative $\partial_{k}\sigBS(t,k)$ can be obtained by differentiating the equation defining the implied volatility $\sigBS$ with respect to the log-moneyness $k$.
More precisely, denoting $C_\mathrm{BS}(k, v)$ the Black--Scholes price of a call option with log-moneyness $k$ and total volatility parameter $v = \sqrt t \, \sigma$, we have  
\be \label{e:implied_vol_def}
\mathbb E \bigl[(S_0 e^{X_t} - S_0 e^k)^+ \bigr] = C_\mathrm{BS}\bigl(k, \sqrt t \, \sigBS(t,k) \bigr) ,
\ee 
for all $k$ and $t$.
Taking the derivative at both sides of \eqref{e:implied_vol_def} with respect to $k$ and using the expressions of the first-order Black--Scholes
greeks $\partial_k C_\mathrm{BS}(k, v)$ and $\partial_v C_\mathrm{BS}(k, v)$, we have
\[
\begin{aligned}
	\partial_{k}\sigBS(t,k)
	&=  \frac{ -\partial_k C_\mathrm{BS}(k, v) - S_{0} e^k \, \mathbb{P} \bigl(X_{t}\geq k \bigr)}
	{\sqrt t \, \partial_{v} C_\mathrm{BS}(k,v)}\bigg|_{v= \sqrt t \, \sigBS(t,k)} \,
	\\
	&= \frac{N\bigl(d_2(k, v)\bigr)  - \mathbb{P} \bigl(X_{t}\geq k \bigr)}
	{\sqrt t \, \phi \bigl(d_2(k, v)\bigr)} \bigg|_{v= \sqrt t \, \sigBS(t,k)} \,,
	\\
\end{aligned}
\]
where $d_2(k, v) = -\frac k v - \frac v 2$, and $\phi$ (resp.\ $N$) denotes the standard Gaussian density (resp.\ cumulative distribution).
The representation above for the implied volatility skew allows us to avoid finite difference methods and only requires us to estimate $\sigBS(t,k)$ and $\mathbb{P}(X_t \geq k)$, which we can do with the same Monte Carlo sample, in order to estimate $\partial_{k}\sigBS(t,k)$ (and therefore, in particular, the at-the-money skew $\partial_{k}\sigBS(t,0)$).

\paragraph{The estimator of the local volatility function.}
Given the Monte Carlo samples $(X^{N,m}_{T}, V^m_T)_{1 \le m \le M}$ of the couple $(X_T^N, V_T)$, the conditional expectation \eqref{e:MarkovProj} defining the local volatility function can be estimated appealing to several different regression methods, see, e.g.,\ \cite{tsybakov2008introduction, henry2019non}.
We have implemented and benchmarked two different estimators: on the one side, a kernel regressor, already applied to evaluate the Markovian projection 
within the celebrated particular calibration algorithm \cite{GHL2012}, and on the other side, 
an alternative estimator based on the explicit knowledge of the conditional law of $(X_t, V_t) |(W_s)_{s \le t}$.

Our kernel regressor  is the Nadaraya--Watson estimator with bandwidth $\delta$, 
\be \label{e:kernel_reg}
\sigDup^2(t, k) 
= \mathbb E \left[V_t |X_t = k\right]
\approx
\frac{ \sum_{m=1}^M V_t^m K_{\delta} \bigl(X_t^{N,m} - k \bigr) } { \sum_{m=1}^M K_{\delta} \bigl(X_t^{N,m} - k \bigr) }.
\ee
We used a Gaussian kernel $K_{\delta}(x)=\exp(-\delta x^{2})$ in our tests.

On the other hand, it is a standard fact that conditionally on $\cF_{t}=\sigma(W_{u}:u\leq t)$, the instantaneous variance $V_t$ is known, while 
the log-price $X_{t}$ is normally distributed with mean $-\frac{1}{2}\int_{0}^{t} V_{s}\dd s + \rho\int_{0}^{t} \sqrt{V_s} \dd W_s$
and variance $(1-\rho^{2})\int_{0}^{t}V_{s}\dd s$. 
This property yields a representation of the Markovian projection $\sigDup(\cdot, \cdot)$ as the ratio of two expectations,
\begin{equation}
	\label{e:LVrepresentation}
	\sigDup^2(t, k)
	= \mathbb E \left[V_t |X_t = k\right]
	= \frac{\mathbb E \left[V_t \, \Pi_t(k) \right] }{ \mathbb E \left[\Pi_t(k)\right]} 
\end{equation}
where
\[
\Pi_t(k)
=  \frac 1 {\sqrt{\int_0^t V_s \dd s}}
\exp \biggl(
- \frac 1 {2 (1 - \rho^2) \int_0^t V_s \dd s }
\Bigl(k + \frac 12 \int_0^t V_s \dd s - \rho \int_0^t \sqrt {V_s} \dd W_s \Bigr)^2
\biggr).
\]
A derivation of \eqref{e:LVrepresentation} can be found in \cite[Proposition 3.1]{Lee2001}; incidentally, this representation of $\sigDup$ has been exploited in 
\cite{henry2009calibration} in the context of a calibration strategy of local stochastic volatility models -- prior to the particular algorithm \cite{GHL2012}.

\paragraph{The estimator of the local volatility skew.}
Differentiating the right-hand side of \eqref{e:LVrepresentation} with respect to $k$, we obtain a representation of $\partial_k \sigDup(t, k)$:
\be \label{e:LV_skew_representation} 
\partial_{k}\sigDup(t,k)
=
\frac{\frac{\partial}{\partial k} \left(\frac{\bE\left[V_{t}\Pi_{t}\right]}{\bE\left[\Pi_{t}\right]}\right)}{2 \, \sigDup(t,k)} 
= \frac{\bE\left[V_{t}\Pi_{t}\right]\bE\left[\frac{U}{\int_{0}^{t} V_{s}\dd s}\Pi_{t}\right]
	- \bE\left[\frac{U}{\int_{0}^{t}V_{s}\dd s}\Pi_{t}V_{t}\right]\bE\left[\Pi_{t}\right]}{2(1-\rho^{2})\bE\left[V_{t}\Pi_{t}\right]^{1/2}\bE\left[\Pi_{t}\right]^{3/2}} \,,
\ee
where $\Pi_{t}$ is a shorthand for $\Pi_{t}(k)$, 
$U = U(k) = k + \frac{1}{2}\int_{0}^{t}V_{s}\dd s - \rho\int_{0}^{t}\sqrt{V_{s}}\dd W_{s}$,
and  $\frac{\partial\Pi_{t}}{\partial k} = -\frac{U}{(1-\rho^{2})\int_{0}^{t}V_{s}\dd s} \,\Pi_{t}$.
All the expectations  appearing in \eqref{e:LVrepresentation} and \eqref{e:LV_skew_representation} can be estimated based on  the exact simulation of the discretized variance path $(V_{t_k})_{1 \le k \le N}$; we approximate the integrals $\int_0^t V_s \dd s$ and $\int_0^t \sqrt{V_s} \dd W_s$ using left-point Euler schemes.
Note that the resulting non-parametric estimators based on the representations \eqref{e:LVrepresentation} and \eqref{e:LV_skew_representation} do not contain any kernel bandwidth or other 
hyper-parameters to be tuned. This is a clear advantage with respect to \eqref{e:kernel_reg}.
We have nevertheless tested both estimators \eqref{e:kernel_reg} and \eqref{e:LVrepresentation} for the Markovian projection function, and found perfect agreement between the two in our tests -- in other words, the local volatilities and local volatility skews computed with the two different methods would be indistinguishable in Figures \ref{fig:ATMskews} and \ref{fig:local_vol_smiles}.
\medskip

In Figure \ref{fig:ATMskews}, we plot the term structure of the ATM implied and local volatility skews, for three different values of $H$ and maturities up to $T=0.5$ years. 
As pointed out in the Introduction and in section \ref{sec:modeling_framework}, the power-law behavior of the ATM implied volatility skew generated by the rough Bergomi model is already well-known; on the other hand, the power-law behavior observed for the local volatility skew in Figure \ref{fig:ATMskews} is (to the best of our knowledge) new, and consistent with Corollary \ref{corollary:Hplus32}.
Figure \ref{fig:skews_ratio} shows the ratio of the implied volatility ATM skew over the local volatility ATM skew, that is the ratio of the curves observed in Figure \ref{fig:ATMskews}, for the different values of $H$: the numerical results are in very good agreement with the ``$\frac{1}{H+3/2}$ rule'' announced in Corollary \ref{corollary:Hplus32}. 
Additionally, we note that the ratio of the two skews seems to be rather stable -- its value is almost constant for maturities up to $T=0.5$ years, with our parameter setup.

\subsection{Short-dated local volatility} \label{s:short_dated_LV}

Theorem \ref{th:asymptotic:markovian:projection} gives the asymptotic behavior of $\sigDup \bigl(T, y \, T^{1/2-H} \bigr)$ as $T$ becomes small. 
Since $y$ is allowed to vary around the at-the-money point $y=0$,  we can check whether the limit \eqref{eq:asy:loc} holds for the function $y \mapsto \hat{\sigma}_{\mathrm{loc}}(T,y) := \sigDup \bigl(T, y \, T^{1/2-H} \bigr)$, that is the whole local volatility smile rescaled with maturity.
The computation of the limiting function $\sigma\bigl( \hat h^y_1 \bigr)$ requires us to evaluate the Cameron-Martin path $h^y \in H^1$ that minimizes the rate function in \eqref{e:rateF}, for given $y$. 
We follow the procedure already exploited in \cite{forde2017asymptotics} and \cite[section 5.1]{friz2021short} : it can be shown, see \cite{forde2017asymptotics}, that the rate function satisfies the alternative representation 
$\Lambda(y) = \inf \Bigl\{\frac{(y - \rho \, G(h))^2}{2 \, \bar \rho^2 F(h)} + \frac 12 \langle \dot{h}, \dot{h} \rangle: \dot{h} \in L^2(0,1) \Bigr\}$, with $F(h) = \langle \sigma^2(\hat h), 1 \rangle = \int_0^1 \sigma^2(\hat h_t) \dd t$ and $G(h) = \langle \sigma(\hat h),  \dot h \rangle = \int_0^1 \sigma(\hat h_t) \dot h_t \dd t$.
This alternative representation yields the rate function under the form of  an unconstrained optimization problem (as opposed to the constrained optimization in \eqref{e:rateF}), which can then be approximately solved by the projection of the one-dimensional path $h$ over an orthonormal basis $\{\dot{e}_n\}_{n \ge 1}$ of $L^2$, $\dot{h}_t = \sum_{n \ge 1} a_n \dot{e}_n(t)$. 
In practice, we truncate the sum at a certain order $N$ and minimize over the coefficients $(a_n)_{1 \le n \le N}$; we obtain an approximation of the minimizer $h^y$ and therefore of $\hat h^y_t = (K^H * \dot h^y)_t = \int_0^t K(t,s) \dot h_s \dd s$.
We chose the Fourier basis $\bigl\{\dot{e}_{1}(t) =1, \ \dot{e}_{2n}(t) = \sqrt 2 \cos(2 \pi n \, t), \ \dot{e}_{2n+1}(t) = \sqrt 2 \sin(2 \pi n \, t), n \in \mathbb N \setminus \{0\}  \bigr\}$ in our experiments, and observed that truncation of the sum at $N = 8$ provides a good accuracy.
The results for the rough Bergomi model are displayed in Figure \ref{fig:local_vol_smiles}, where the function $\hat{\sigma}_{\mathrm{loc}}(T, y)$ is indeed seen to approach its limit $\sigma\bigl( \hat h^y_1 \bigr)$ when maturity decreases from $T=0.5$ to $T=0.05$.
The residual error term $\hat{\sigma}_{\mathrm{loc}}(T, y) - \sigma\bigl( \hat h^y_1 \bigr)$ is seen to depend on $H$, with lower values of $H$ being associated with higher errors.
It is however unclear whether the error for $H = 0.1$ is due to the slow convergence of $\hat{\sigma}_{\mathrm{loc}}$ or the weak error rate due to the Monte Carlo simulation (see Remark \ref{rem:weak_error_rate}).

\paragraph{Extrapolation of local volatility surfaces.}
Eventually, Theorem \ref{th:asymptotic:markovian:projection} provides us with an extrapolation recipe of local volatilities for very short maturities: fixing a (small) maturity $T$ and a log-moneyness level $k$, formally plugging $y =  \frac k {T^{1/2-H}}$ in \eqref{eq:asy:loc}  we obtain
\[
\sigDup \bigl(T, k \bigr) \approx  
\sigma \Bigl( \hat h_1^{y} \Bigr)\Bigr|_{y =  \frac k {T^{1/2-H}}} \,.
\]
The limiting function $\sigma \bigl( \hat h_1^{y} \bigr)|_{y =  \frac k {T^{1/2-H}} }$ can therefore be used to extrapolate a local volatility surface in a way that is consistent with the behavior implied by a rough volatility model.

As a specific application, consider the calibration of a local-stochastic volatility model (LSV) to an option price surface, for example using the particle method of Guyon and Henry-Labord{\`e}re \cite{GHL2012}. The LSV model can be obtained by the decoration of a naked rough volatility model, which amounts to enhancing the rough volatility model \eqref{eq:rvol:model} for $S_t = S_0 e^{X_t}$ with a leverage function $l(t,S)$,
\[
\dd S_t = S_t \, l(t,S_t) \sqrt{V_t} \left( \rho \, \dd W_t + \sqrt{1 - \rho^2} \, \dd \overline W_t \right) \,.
\]
Given the spot variance process $V$, the LSV model calibrated to a given Dupire local volatility surface $\sigma_\mathrm{Dup}$, corresponds to (see \cite{GHL2012})
\[
l(t,S_t) = \frac{ \sigma_\mathrm{Dup}(t,S_t) }  {\sqrt{ \mathbb E[V_t | S_t]}}.
\] 
In general, one wishes the leverage function $l(t,S)$ to be a small correction to the original stochastic volatility model (in other words: as close as possible to $l \equiv 1$).
In practice, the local volatility $\sigma_\mathrm{Dup}$ coming from market data has to be extrapolated  for values of $t$ smaller than the shorter observed maturity, and the choice of the extrapolation method is up to the user.
If, for small $t$, the chosen extrapolation $\sigma_\mathrm{Dup}(t,K)$ is qualitatively  too different from the behavior of the conditional expectation $\mathbb E[V_t | S_t = K]$ in the rough volatility setting (for example, more specifically: the ATM skew of $\sigma_\mathrm{Dup}$ is far from the power law \eqref{eq:ldp:loc:vol:skew}), then the leverage function will have to compensate, hence deviating from the unit function.
Under the pure rough volatility model ($l \equiv 1$), Theorem \ref{th:asymptotic:markovian:projection} and Corollary \ref{corollary:Hplus32} describe the behavior of  the Markovian projection $\mathbb E[V_t | S_t]$ for small $t$: eventually, these statements give hints on how $\sigma_\mathrm{Dup}(t,\cdot)$ should be extrapolated for $l(t,\cdot) $ not to deviate too much from the unit function.
Such an extrapolation scheme is exploited in the recent work of Dall'Acqua, Longoni and Pallavicini \cite{DallAcqua2022}, precisely in order to calibrate a LSV model with lifted Heston \cite{AbiJaber2019} backbone to the implied volatility surface of the EuroStoxx50 index.\footnote{We thank Andrea Pallavicini and Riccardo Longoni for interesting and stimulating discussions on this topic.}

\paragraph{Failure of the harmonic mean asymptotic formula under rough volatility.}
In Remark \ref{rem:harm_mean}, we pointed out that, as a consequence of the general $\frac{1}{H+3/2}$ skew rule (as opposed to the $1/2$ rule) in Corollary \ref{corollary:Hplus32}, the harmonic mean asymptotic formula  $\sigBS(T, k) \sim H(T,k)$ as $T \to 0$, see \eqref{eq:def:harmonic:mean}, is expected not to hold for $H \neq 1/2$ (without any contradiction with the statements in \cite{BBF02}, which require regularity conditions on local volatility surface that are not satisfied in the rough volatility setting, see our discussion in Remark \ref{rem:harm_mean}).
In other words, we do not expect the harmonic mean of the local volatility $H(T,k) = \frac 1{\frac 1k \int_0^k \frac {\dd y}{\sigDup(T,y)}}$ to be a good approximation of the implied volatility when maturities become small when the involved volatility surfaces are generated by a rough vol model.
Having constructed estimators \eqref{e:kernel_reg}
 and \eqref{e:LVrepresentation} of the local volatility function under the rough Bergomi model, we are also able to approximate (with an additional deterministic quadrature) the harmonic mean $H(T,k)$, and compare the output with the implied volatility smile.
The results are shown in Figure \ref{fig4}, for three different  values of $H$.
As expected, when $H=0.5$ we observe (upper left panel) that the implied volatility $\sigBS(T,k)$ approaches the harmonic mean $H(T,k)$ when maturity decreases from $T=0.45$ to $T=0.05$, and the at-the-money slopes are also seen to agree. The convergence is even more apparent in the upper right figure, where the ratio $\frac{\sigBS(T,k)}{H(T,k)}$ is seen to monotonically converge to one.
This behavior should be compared with the one in the two bottom figures, where the rough case $H=0.1$ is considered (the case $H=0.3$ being intermediate between the other two): now, when maturity decreases, the implied volatility smile does not seem to approach the harmonic mean $H(T,k)$ anymore (apart from the specific at-the-money point $k=0$ where both functions tend to the initial spot volatility $\sigma_0 = \sqrt{V_0}$), and in particular, the slopes of the two curves are seen to considerably deviate from each other. This phenomenon is even more clear in the bottom right figure, where the ratio $\frac{\sigBS(T,k)}{H(T,k)}$ has a completely different behavior with respect to the diffusive case $H=0.5$.

\bibliographystyle{abbrv}
\bibliography{roughvol}

\begin{figure}[H]
	\centering
	\includegraphics[scale=0.48]{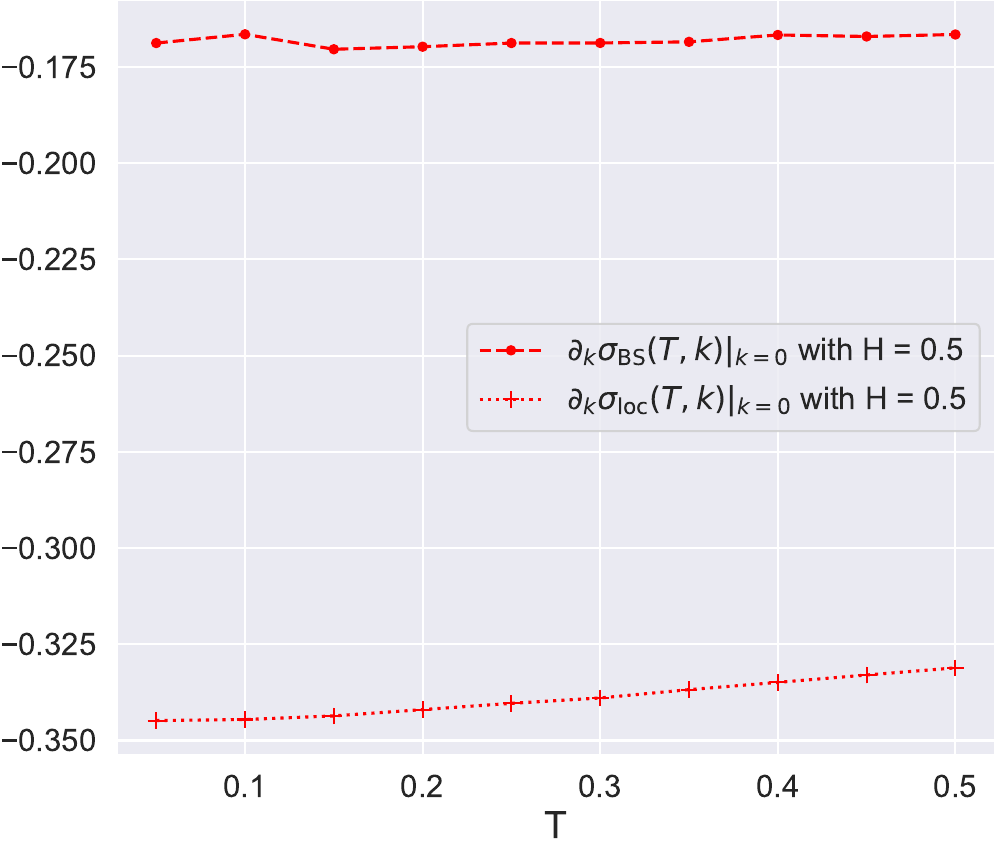}
	\includegraphics[scale=0.48]{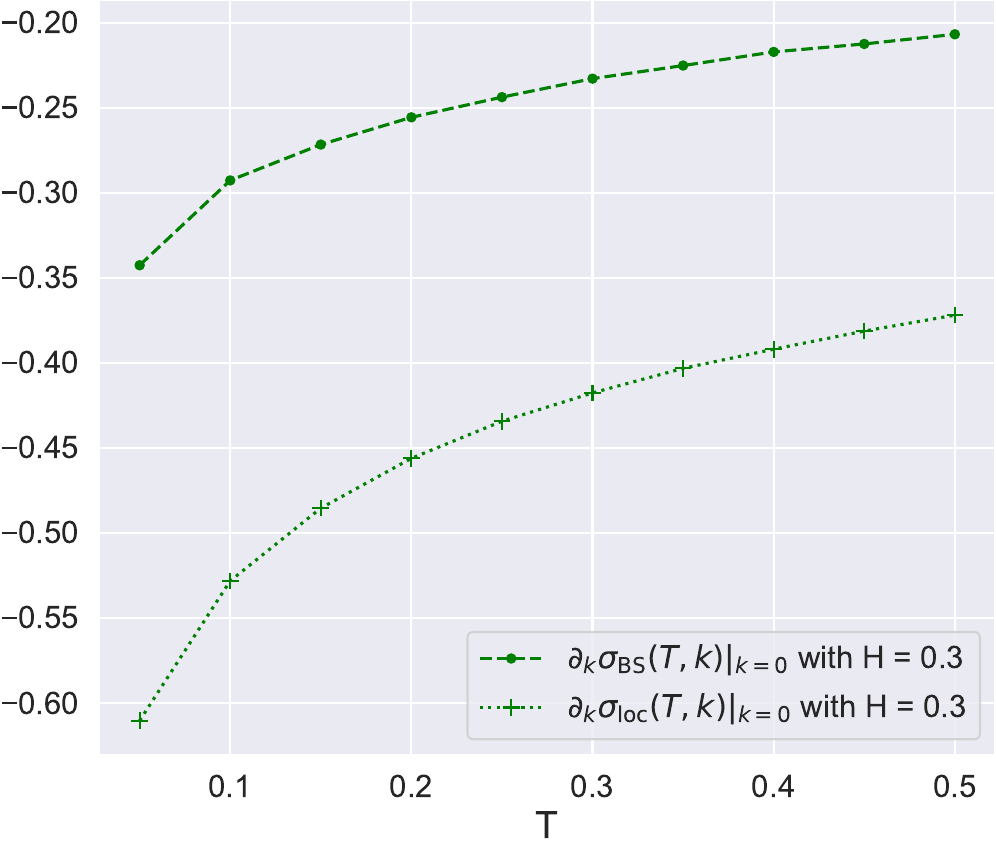}\\
	\includegraphics[scale=0.48]{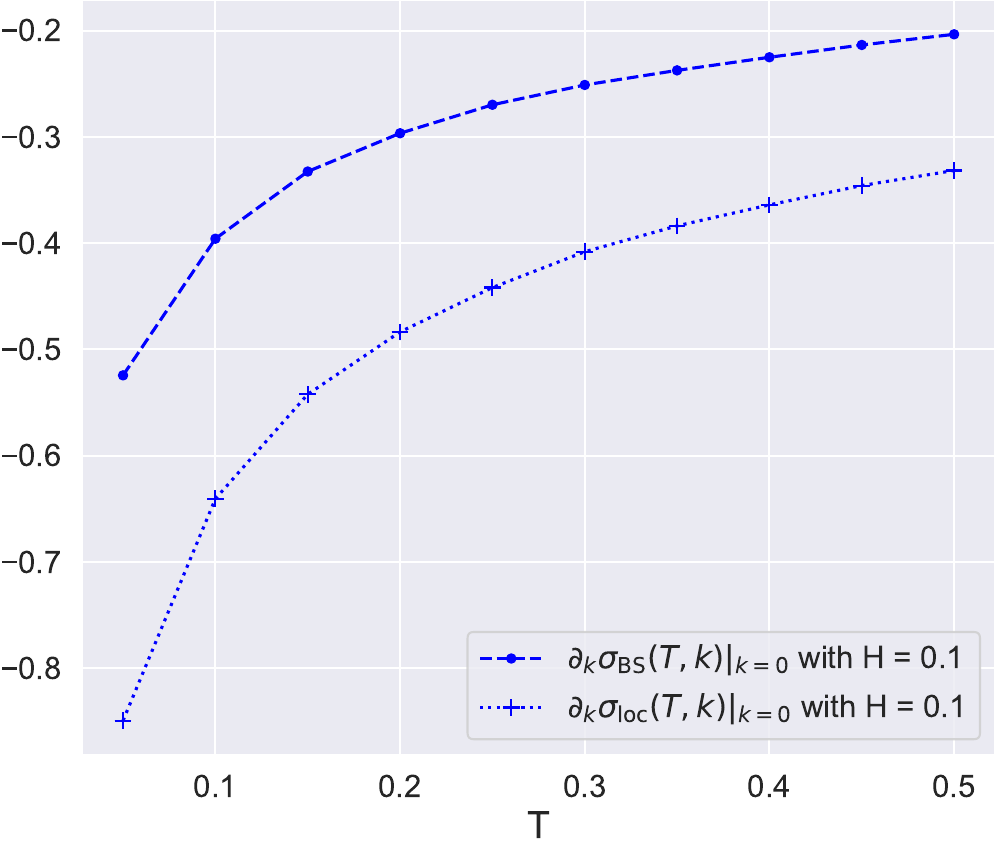}
	\caption{At-the-money implied and local volatility skews in the rough Bergomi model \eqref{e:rBergomi} for $H=0.5$ (red, top left figure),  $H=0.3$ (green, top right figure), and $H=0.1$ (blue, bottom figure).
		The maturity $T$ on the $x$-axis is expressed in years.}
		\label{fig:ATMskews}
\end{figure}

\begin{figure}[H]
	\begin{center}
		\includegraphics[scale=0.6]{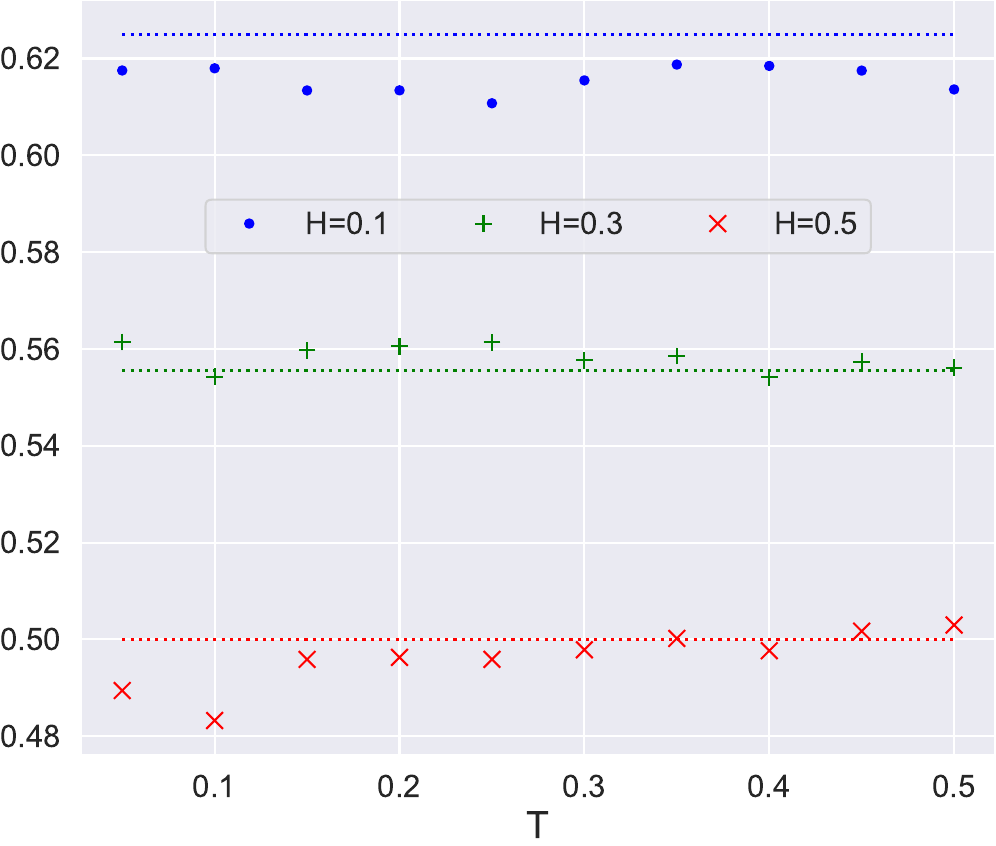}
	\end{center}
	\caption{Numerical evidence for the $\frac{1}{H+3/2}$ ratio rule stated in Corollary \ref{corollary:Hplus32}: we plot the ratio of the at-the-money implied and local volatility skews  $\frac{\partial_k \sigBS (T,k)|_{k=0}}{\partial_k \sigDup (T,k)|_{k=0}}$  for $H \in \{0.1, 0.3, 0.5\}$ against maturity $T$ (in years).
		The dashed lines correspond to the constant values $\frac{1}{H+3/2}$ (blue for $H=0.1$, green for $H=0.3$, red for $H=0.5$).
	}
\label{fig:skews_ratio}
\end{figure}

\begin{figure}[H]
	\centering
	\includegraphics[scale=0.48]{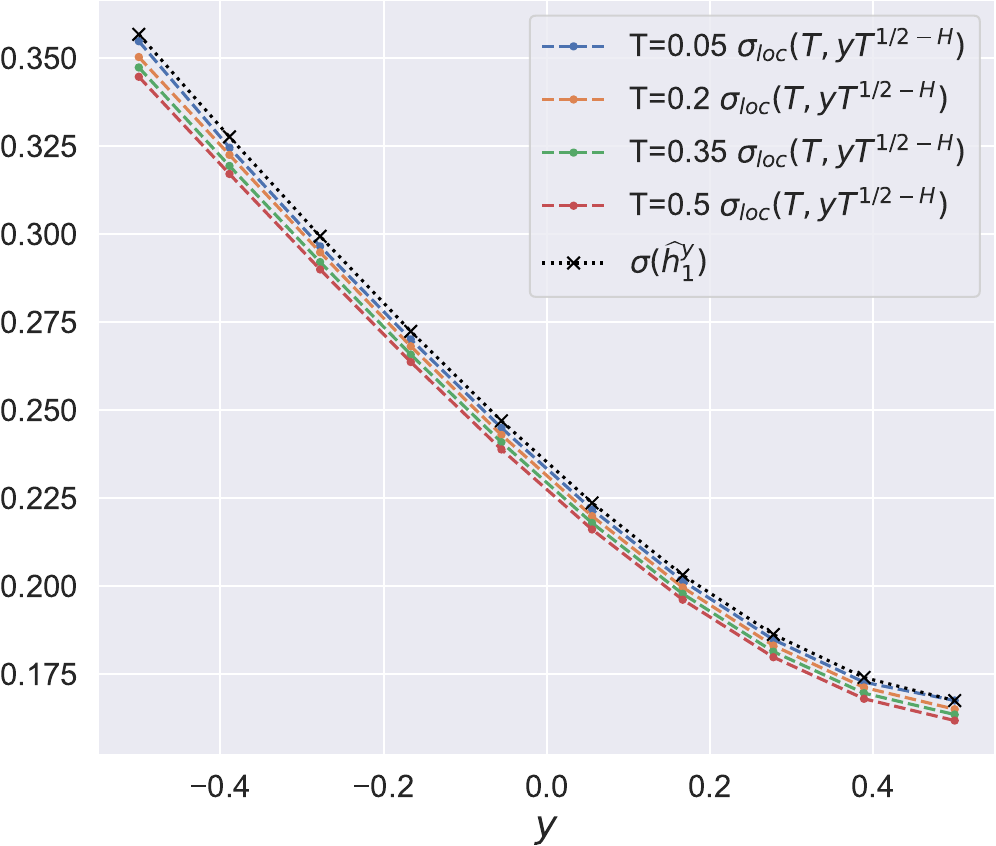}
	\includegraphics[scale=0.48]{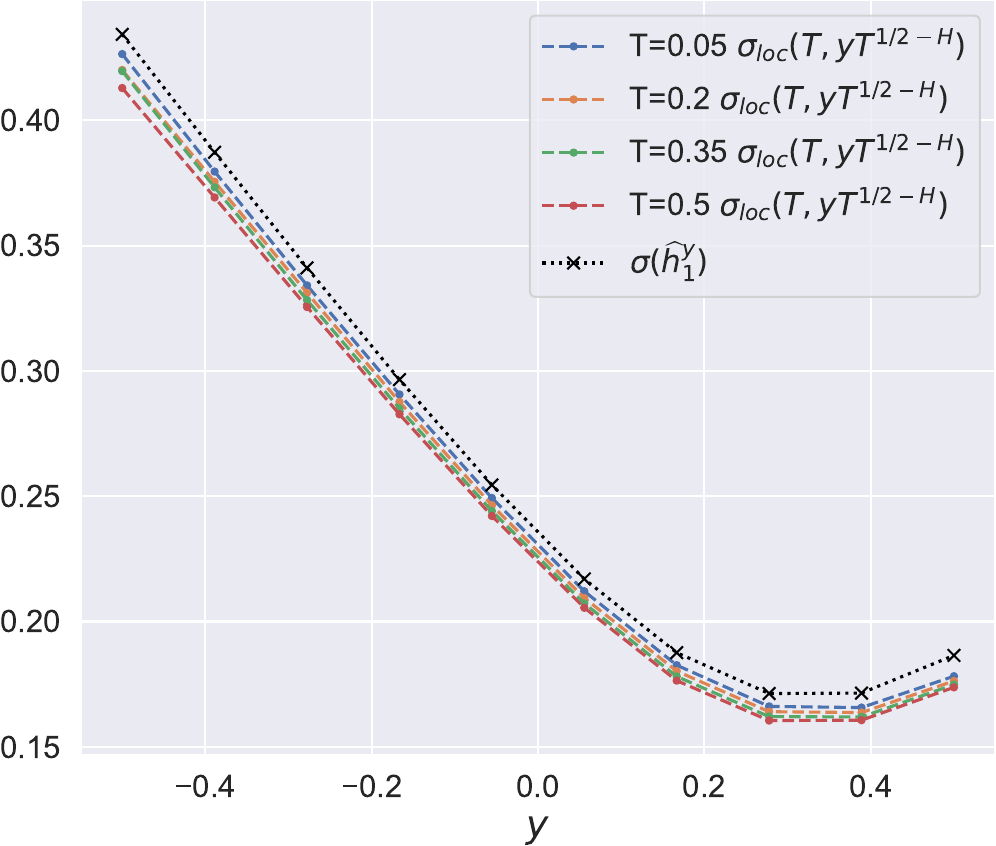}\\
	\includegraphics[scale=0.48]{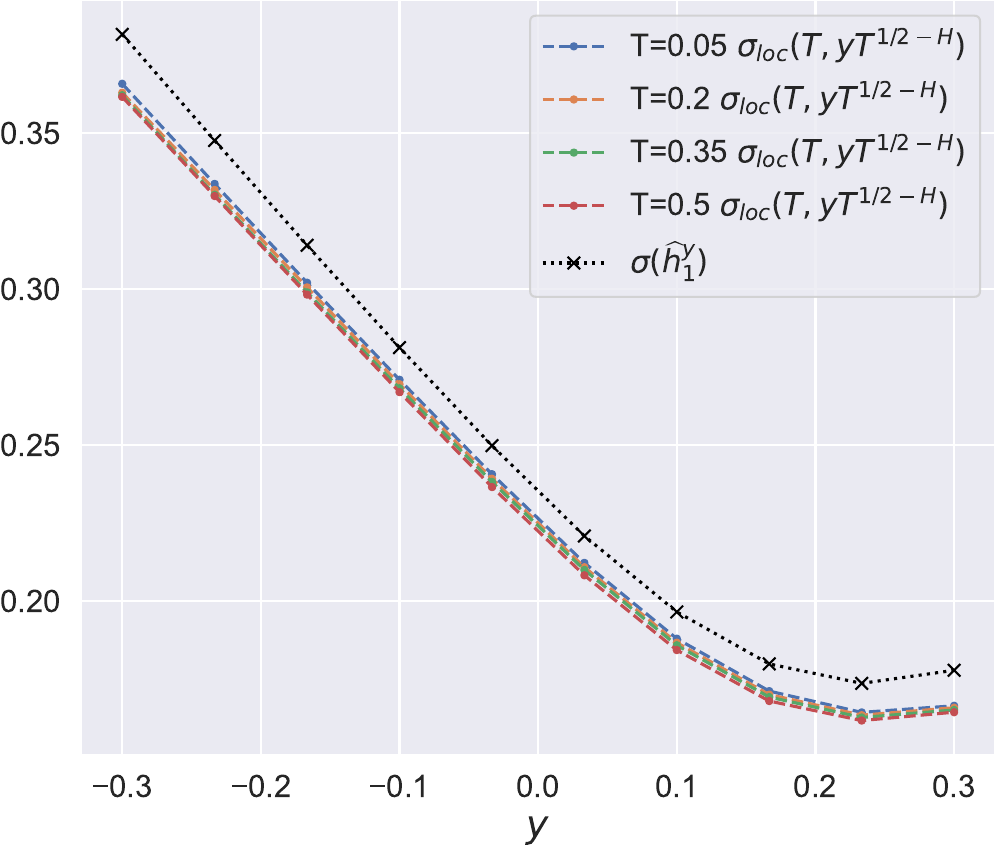}
	\caption{Short-dated local volatility  the rough Bergomi model \eqref{e:rBergomi} for $H=0.5$ (top left figure), $H=0.3$ (top right figure), and $H=0.1$ (bottom figure).
		Recall that, according to Theorem \ref{th:asymptotic:markovian:projection}, $\sigDup(T, y \, T^{1/2 - H}) \to \sigma(\hat h^y_1)$ as $T \to 0$.
		The rate function minimizing path $\hat h^y_t$ is evaluated using the Ritz projection method with $N = 8$
		Fourier basis functions, see section \ref{s:short_dated_LV}.}
		\label{fig:local_vol_smiles}
\end{figure}

\begin{figure}[H]
	\centering
	\includegraphics[scale=0.45]{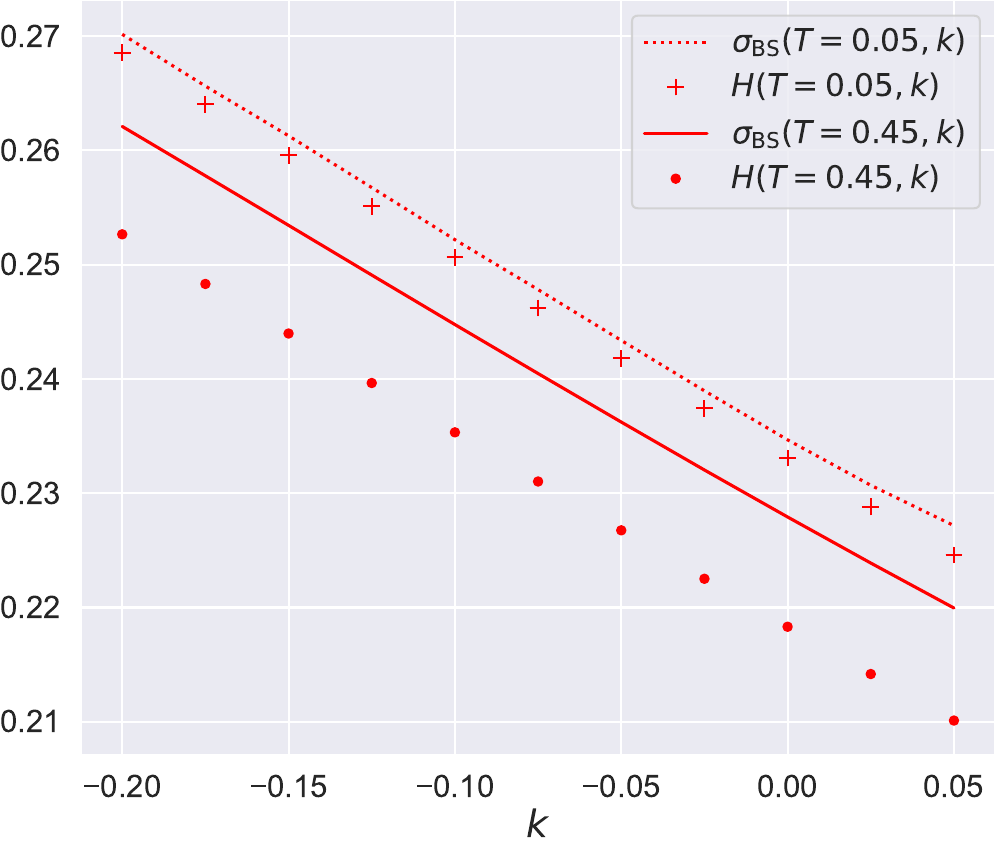}
	\includegraphics[scale=0.45]{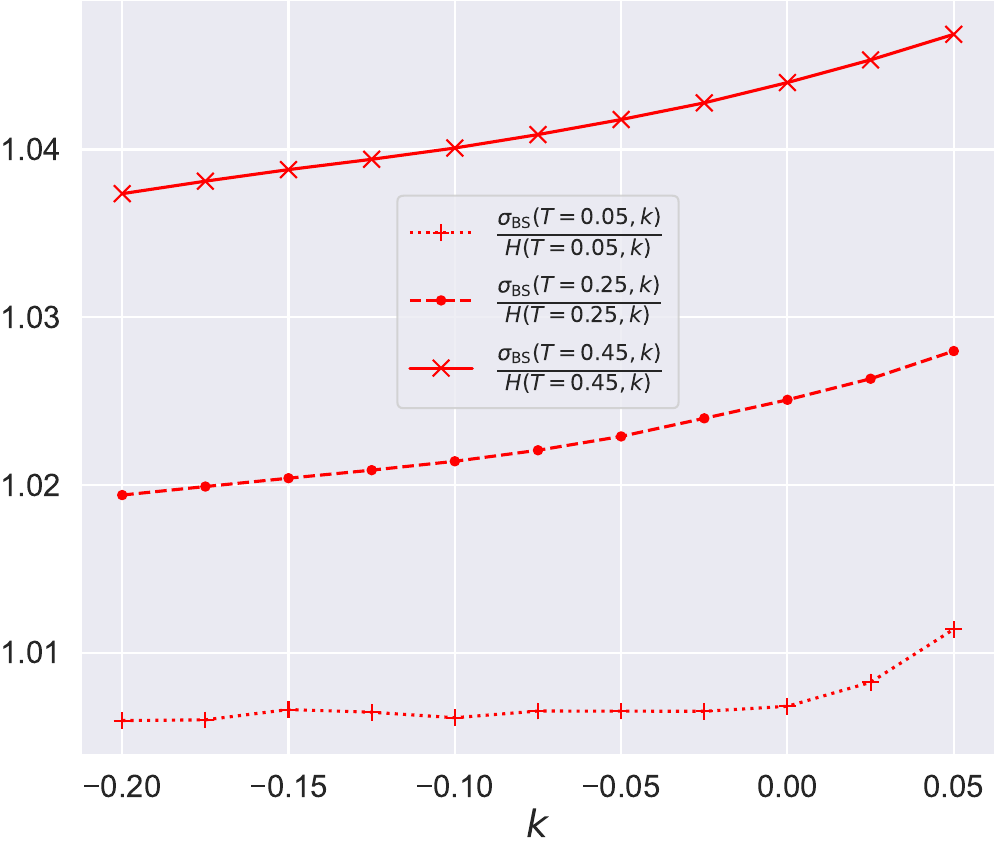}\\
	\includegraphics[scale=0.45]{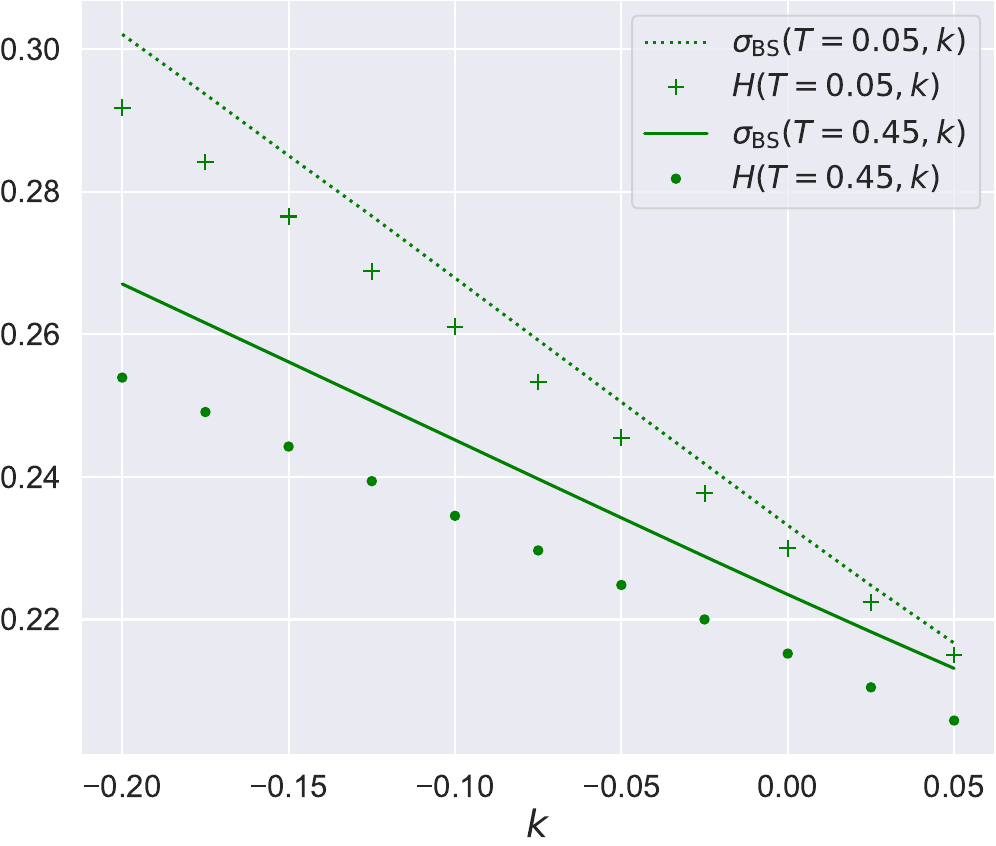}
	\includegraphics[scale=0.45]{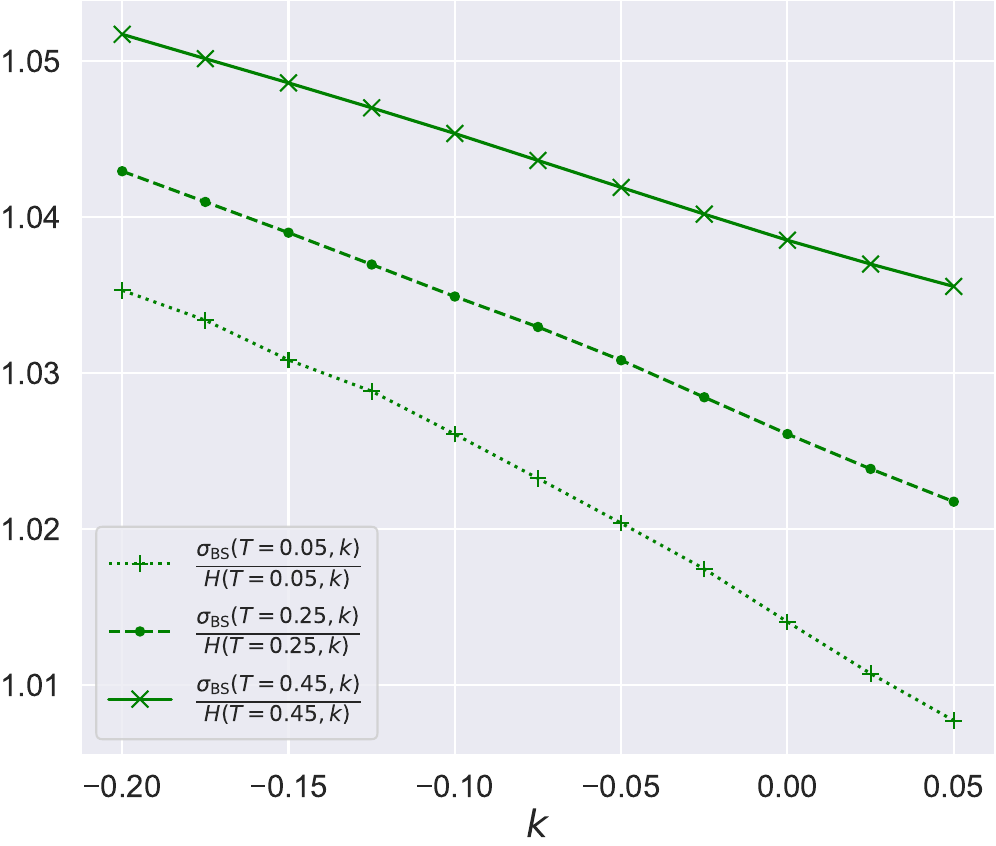}\\
	\includegraphics[scale=0.45]{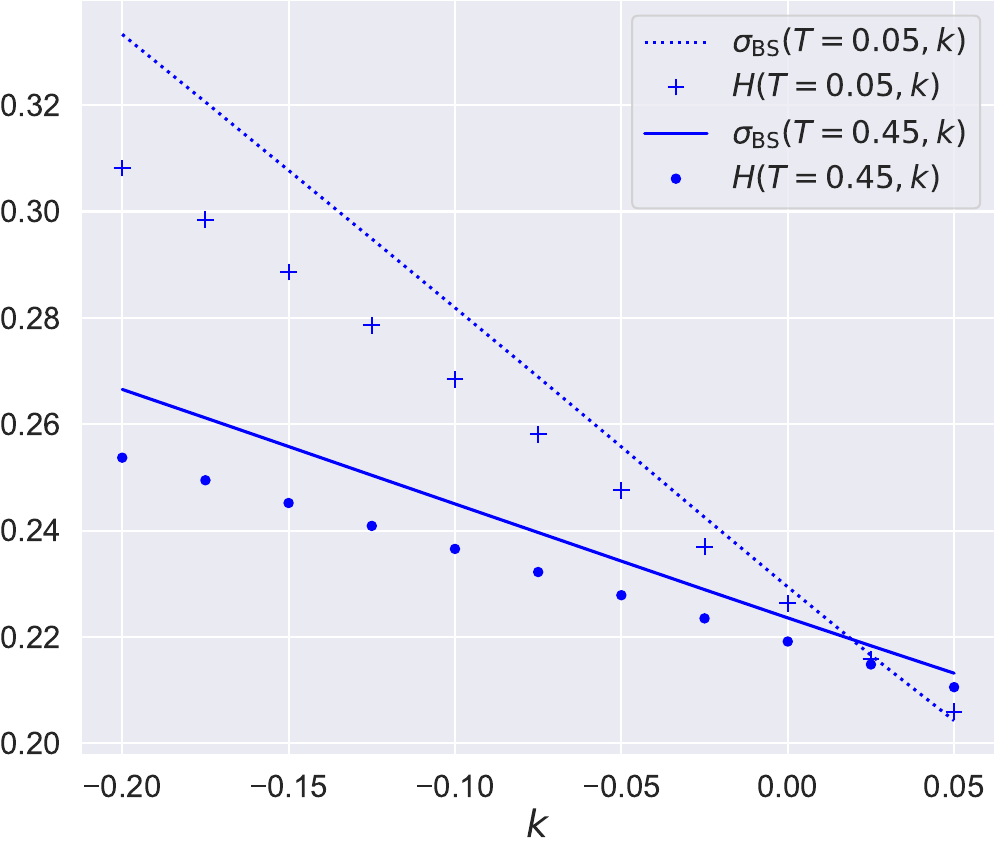}
	\includegraphics[scale=0.45]{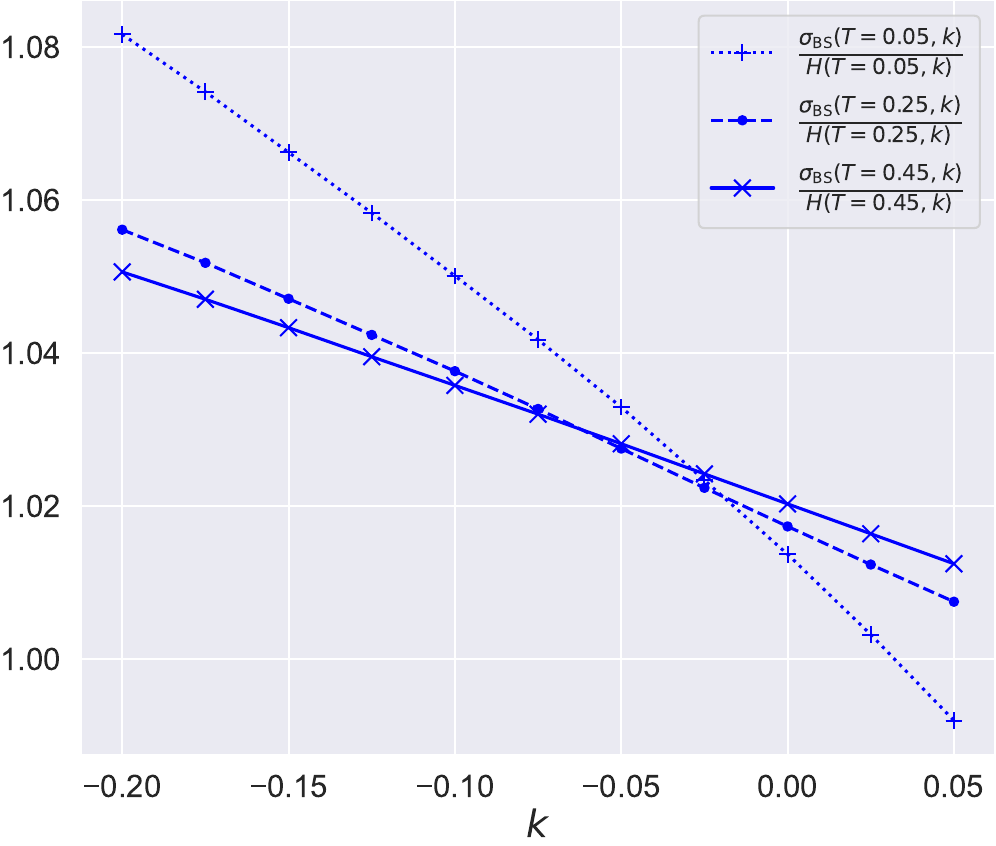}
	\caption{Numerical evidence for the failure of the harmonic mean formula within the rough Bergomi model \eqref{e:rBergomi} (see Remark \ref{rem:harm_mean}): 
		in the left figures, we compare the  implied volatility $\sigBS(T, k)$ and the harmonic mean $H(T, k)$ of the local volatility defined in \eqref{eq:def:harmonic:mean},
		for two different maturities $T$ 
		and for $H=0.5$ (red), $H=0.3$ (green), and $H=0.1$ (blue).
		In the right figures (same color conventions as the left figures), we plot the ratio $\frac{\sigBS(T, k)}{H(T,k)}$ of the two functions, expected to tend to $1$ as $T
		\to 0$ when $H=0.5$.}
		\label{fig4}	
\end{figure}

\section{Proofs}\label{sec:proofs}
\begin{proof}[Proof of Corollary \ref{corollary:Hplus32}]
Equation \eqref{eq:ldp:loc:vol:skew} is a straightforward consequence of Theorem \ref{th:asymptotic:markovian:projection}.

It is a standard result that the implied volatility $\sigBS$ and the local volatility $\sigDup$ generated by a stochastic volatility model with 
$\rho = 0$ are symmetric around $y = 0$, so that the finite-difference at-the-money skews $\Sbs$ and $\Sloc$ are identically zero in this case.
We, therefore, assume $\rho \neq 0$ in what follows. Let us write $\langle K1,1\rangle= \int_0^1 K1(t) \dd t$ and $K1(t)=\int_0^t K(t,s) \dd s$.
Using an expansion of the map $y \mapsto \hat{h}^y_1$ around $y = 0$ as provided in \cite{friz2021short},  we have
\[
\Sigma(y)= \sigma \bigl(\hat{h}^y_1\bigr)
= \sigma_0 + y \frac{\sigma_0'}{\sigma_0} \rho K1(1) + \cO(y^2) 
\qquad \mbox{ as } y \to 0 \,.
\]
Together with \eqref{eq:ldp:loc:vol:skew}, this implies
\[
\Sloc(t, y) \sim
 \Bigl(\frac{\sigma_0'}{\sigma_0} \rho K1(1)
 + r(y) \Bigr)  \frac 1 {t^{1/2 - H}}
\]
as $t \to 0$, where $r(y) \to 0$ as $y \to 0$. 
Similarly, from \eqref{eq:asy:bs} and a third order energy expansion of $\Lambda$, obtained in \cite[Thm 3.4]{bayer2017short} (an extension to forth order is given in \cite{friz2021short} but not required here) it follows that
\[
\Sbs (t, y) \sim
\frac{\chi(y) - \chi(-y)} {2y} 
 \frac 1 {t^{1/2 - H}}
=
 \Bigl(\frac{\sigma_0'}{\sigma_0} \rho \langle K1,1 \rangle + \ell(y) \Bigr)  \frac 1 {t^{1/2 - H}}
\]
as $t \to 0$, where $\ell(y) \to 0$ as $y \to 0$.
Therefore
\[
\frac{\Sbs(t, y)}{\Sloc(t, y)}
\to
\frac{\chi(y) - \chi(-y)}
{\Sigma(y) - \Sigma(-y)} = \frac{ \langle K1,1 \rangle}{K1(1)} + o(1)
\qquad \text{ as  } t \to 0.
\]
The identity
\[
K1(1)=	(H+3/2)
\langle K1,1\rangle
\]
for $K(t,s)=\sqrt{2H}(t-s)^{H-1/2}$ is straightforward to prove using simple integration
(we note in passing that this identity holds for any self-similar $\hat{W}$, by leveraging a representation in \cite{jost2007}). The statement of the corollary follows.
\end{proof}

The proof of Theorem \ref{th:asymptotic:markovian:projection} is based on the following representation of the Markovian projection, 
based on the integration by parts of the Malliavin calculus,
\be \label{e:condExpRepr}
\bE [V_T |X_T = y ]
=
\frac{ 
\bE \Bigl[ V_T 
\1_{X_T \geq y }
\int_0^T\frac{1}{\bar{\rho} \sigma(\hat{W}_t)  } \dd \bar{W}_t  \Bigr]
}{
\bE \Bigl[ \1_{X_T \geq y} \int_0^T\frac{1}{\bar{\rho} \sigma(\hat{W}_t) } \dd \bar{W}_t \Bigr]
}.
\ee
Representation \eqref{e:condExpRepr} is rather classical, see \cite{fournie2001applications}, though spelled out only in case $\rho = 0$, and \cite[Lemma 3]{bally:inria-00071868} for a general formula in an abstract setting. We note that $\bar \rho$ cancels as long as it is not zero, equivalently $|\rho| < 1$, which is our non-degeneracy assumption. (We kept $\bar\rho$ above to insist on this point.) 
For completeness, we give a proof in Lemma \ref{lemma:cond:exp}.

\begin{proof}[Proof of Theorem \ref{th:asymptotic:markovian:projection}]
Setting $T = \eps^2$ and using the
time-scaling property of the triple $(\bar W, \hat W, X)$, we get
\[
\begin{split}
\bE \Bigl[ \sigma(\hat{W}_T)^2 \1_{X_T \geq y T^{1/2-H}} \int_0^T\frac{1}{\bar{\rho} \sigma(\hat{W}_t)  } \dd \bar{W}_t \Bigr]
&=
\bE \Bigl[\sigma(\eps^{2H}\hat{W}_1)^2 \1_{X_1^\eps \geq y \eps^{1-2H}} \int_0^1\frac{1}{\bar{\rho} \sigma(\eps^{2H} \hat{W}_t) } \eps  \dd \bar{W}_t\Bigr]
\\
\bE \Bigl[ \1_{X_T \geq y T^{1/2-H}} \int_0^T\frac{1}{\bar{\rho} \sigma(\hat{W}_t)  } \dd \bar{W}_t \Bigr]
&=
\bE \Bigl[\1_{X_1^\eps \geq y \eps^{1-2H}} \int_0^1\frac{1}{\bar{\rho} \sigma(\eps^{2H} \hat{W}_t) } \eps  \dd \bar{W}_t \Bigr] .
\end{split}
\]
We define
\begin{equation}\label{eq:def:J}
\begin{split}
J(\eps,y) 
&= 
e^{\frac{\Lambda (y)}{\eps^{4H}}} \bE\Bigr[\sigma(\eps^{2H}\hat{W}_1)^2 \1_{X_1^\eps \geq y \eps^{1-2H}} \int_0^1\frac{1}{\bar{\rho} \sigma(\eps^{2H} \hat{W}_t) } \eps  \dd \bar{W}_t\Big] 
\\
\overline{J}(\eps,y) &= e^{\frac{\Lambda (y)}{\eps^{4H}}} \bE \Bigl[ \1_{X_1^\eps \geq y \eps^{1-2H}}  \int_0^1\frac{1}{\bar{\rho} \sigma(\eps^{2H} \hat{W}_t) } \eps  \dd \bar{W}_t\Big] 
\end{split}
\end{equation}
so that $\sigDup^2 \bigl(\eps^2, y \, \eps^{1-2H} \bigr) = \frac{J(\eps, y)}{\overline{J}(\eps,y)}$ from \eqref{e:condExpRepr}.
The implementation of an infinite-dimensional Laplace method  along the lines of \cite{friz2018precise1} allows us to determine the asymptotic behavior of $J(\eps, y)$ and $\overline{J}(\eps,y)$ as $\eps \to 0$: we postpone the details to Lemma \ref{th:limit:J} below. 
We obtain the $\eps \to 0$ limit of $\sigDup^2 \bigl(\eps^2, y \, \eps^{1-2H} \bigr)$, and therefore the statement of the Theorem, from \eqref{eq:limit:J}.
\end{proof}

\begin{lemma}\label{lemma:cond:exp}
The  representation formula \eqref{e:condExpRepr} for the conditional expectation holds for every $y \in \R$. 
\end{lemma}
\begin{proof}
The Malliavin derivative $\bar{D}$ of $X_T$ with respect to $\bar{W}$ is
\[
\bar{D}_t X_T = \bar{\rho} \sqrt{V}_t,
\qquad t \mino T \,,
\]
because $V$ is $W$-adapted. Consider a two-dimensional Skorohod integrable process $(0,\bar{u})$, with
\[
\bar{u}_t=\frac{1}{T \bar{D}_t X_T}
=
\frac{1}{T  \bar{\rho} \, \sqrt{V}_t}.
\]
We write $\delta$ for the Skorohod integral.
From Malliavin integration by parts formula for a bounded smooth function $\phi: \R \to \R$ and using that $\bar{D}_t V_T=0$,
we obtain
\[
\begin{aligned}
\bE
[V_T \phi(X_T) \delta(0,\bar{u})]
&=
\bE
[\langle \bar{D}(V_T \phi(X_T)),\bar{u}\rangle]
=
\bE
[V_T \phi^{\prime}(X_T) \langle \bar{D}X_T,\bar{u}\rangle]
\\
&= 
\bE
\Big[
V_T \phi^{\prime}(X_T) \int_0^T \bar{D}_t X_T \, \bar{u}_t \dd t
\Big]
=
\bE
\big[
V_T \phi^{\prime}(X_T)
\big].
\end{aligned}
\]
We have used here boundedness of $\phi(\cdot)$ and Assumptions \ref{assu:sigma} on $\sigma(\cdot)$  (see proof of next Lemma \ref{prop:algebraic:equivalence} for a detailed argument).
Moreover, since $V$ is adapted, 
we have $\delta(0,\bar{u})=\int_0^T \frac{1}{T  \bar{\rho} \sqrt{V}_t} \dd \bar{W}_t$, and so
\[
\bE
[V_T \phi^{\prime}(X_T)]
=
\frac{1}{T  \bar{\rho} }
\bE
\Big[V_T \phi(X_T) \int_0^T \frac{1}{\sqrt{V}_t} \dd \bar{W}_t
\Big]
\]
Following the same steps, one can show that the following identity also holds:
\[
\begin{split}
\bE
[
\phi^{\prime}(X_T)
]=
\frac{1}{\bar{\rho} T}
\bE
\Big[
\phi(X_T) \big( \int_0^T\frac{1}{ \sqrt{V_t}} \dd \bar{W}_t\big)
\Big].
\end{split}
\]
The representation formula \eqref{e:condExpRepr} for the conditional expectation then follows from a regularization procedure of the indicator function $\1_{X_t \ge y}$, see for example \cite{bally:inria-00071868}.
\end{proof}

\begin{lemma} \label{lem:exact_asymptotics}
For $y\in\R\setminus\{0\}$ {small enough}, one has $\int_0^1 \frac{ \dd \bar{h}^y_t}{ \sigma \left(\hat{h}^y_t \right) } \neq 0$.
Moreover, for $J,\overline{J}$ defined in \eqref{eq:def:J} we have
\label{th:limit:J}
\begin{equation}\label{eq:limit:J}
\begin{split}
J(\eps,y)&\sim \eps^{2H}
\sigma^2 (\hat{h}_1^y)
\left( \int_0^1\frac{\dd \bar{h}_t^y}{\sigma(\hat{h}_t^y) } \right)\
{ \frac{1}{\bar{\rho} \sqrt{2\pi} \sqrt{2 \Lambda(y)} } }
 \bE \left[ \exp\left(\Lambda^{\prime}
 \left( y\right)\Delta_{2}\right) \right],
\\
\overline{J}(\eps,y)&\sim
\eps^{2H}
\left( \int_0^1\frac{\dd \bar{h}_t^y}{\sigma(\hat{h}_t^y) } \right)\
{ \frac{1}{\bar{\rho} \sqrt{2\pi} \sqrt{2 \Lambda(y)} } }
 \bE \left[ \exp\left(\Lambda^{\prime}\left( y\right)\Delta_{2}\right) \right]
\end{split}
\end{equation}
as $\eps \to 0$, where $\Delta_2$ is a quadratic Wiener functional given in 
\eqref{dec:g2} (see also \cite[Equation (7.4)]{friz2018precise1}).
\end{lemma}

\begin{corollary}[Digital expansion] We do not use it here but we note that from the computations in  the proof of Theorem \ref{th:asymptotic:markovian:projection} and Lemma \ref{lem:exact_asymptotics},	it follows that there exists a $y_0>0$ such that the following holds for all $y \in (0,y_0)$:
\[
\bP \left( X_T \geq y T^{1/2-H} \right) \sim 
e^{-\frac{\Lambda (y)}{T}}
T^{H}
{ \frac{1}{\sqrt{2\pi}\sqrt{2 \Lambda(y)} } }
 \bE \left[ \exp\left(\Lambda^{\prime}
 \left( y\right)\Delta_{2}\right) \right],
\qquad
\text{ as  } T \to 0.
\]
\end{corollary}

\begin{proof}[Proof of Lemma \ref{lem:exact_asymptotics}]
We aim to apply the asymptotic results in \cite{friz2018precise1}.
Assumption (A1) in \cite{friz2018precise1} is nothing but the validity of the large deviations principle for the model defined in \eqref{eq:rvol:model}, which we have already discussed in Remark \ref{rm:condLDP}.
We take $y$ close enough to $0$ so that the non-degeneracy assumptions \cite[Assumptions (A3), (A4), (A5)]{friz2018precise1} are satisfied for the model under consideration, as it has been checked in \cite[Section 7.1]{friz2018precise1}. Therefore, we have that the preliminary regularity structures results in \cite{friz2018precise1} apply to our setting, and can employ them in the proof. 

Using \cite[Proof of Lemma 3.4, Step 1]{friz2021short} we have
\[
\int_0^1  g_t \dd \bar{h}^y_t
=\bar{\rho} \Lambda^{\prime}(y)
\int_0^1 g_t \sigma(\hat{h}^y_t)  \dd t,
\]
for any square-integrable test function $g$,
so that in particular
\[
\int_0^1 \frac{ \dd \bar{h}^y_t}{ \sigma(\hat{h}^y_t) }
=\bar{\rho} \Lambda^{\prime}(y)
\int_0^11 \dd t
=
\bar{\rho} \Lambda^{\prime}(y)
\]
and $\Lambda^{\prime}(y)\neq 0$, as detailed in the proof of \cite[Theorem 6.1]{friz2018precise1}. The  proof of \eqref{eq:limit:J} is then a modification of \cite[Proposition 8.7]{friz2018precise1}, from which we borrow the notations. Necessary definitions are recalled in Appendix \ref{appendix:reg:str}. We only prove the statement for $J$,  the one for $\overline{J}$ being completely analogous.
Set, for any $\delta>0$, 
\begin{equation}\label{def:P:delta}
            \bP_\delta (A) = \bP ( A \cap \{  \eps^{2H} ||| \WW ||| < \delta     \} ),
\end{equation}
with $\WW$ defined in \eqref{def:model}, and set
\[
\begin{split}
J_\delta(\eps,y)
= \eps^{1-2H}
e^{\frac{\Lambda(y)}{\eps^{4H}}} 
\bE_\delta \Big[\sigma(\eps^{2H}\hat{W}_1)^2 \1_{X_1^\eps \geq y \eps^{1-2H}} \big( \int_0^1\frac{1}{\bar{\rho} \sigma(\eps^{2H} \hat{W}_t) } \eps^{2H}
\dd \bar{W}_t\big)\Big]
\end{split}
\]
where the expectation $\bE_\delta$ is with respect to the sub-probability $\mathbb{P}_\delta$.
As a consequence of Lemma \ref{prop:algebraic:equivalence},
any ``algebraic expansion'' of $J$ (i.e., in powers of $\eps$) does not change by switching to $J_\delta$. 
So, proving the asymptotic behavior \eqref{eq:limit:J} for $J_\delta$ implies the statement. 

We recall \eqref{eq:rescX} and apply Girsanov's theorem, via the transformation 
\begin{equation}
\begin{split}
\eps^{2H} \W &\to \eps^{2H} \W + \h^y = \eps^{2H} (\W + \h^y / \eps^{2H}),\\
\eps^{2H} \hat{W} &\to \eps^{2H} \hat{W} + \hat{h}^y = \eps^{2H} (\hat{W} + \hat{h}^y / \eps^{2H})
\end{split}
\end{equation}
from which we introduce
\begin{equation}\label{def:Z}
\bar{Z}_1^\eps=\int_0^1\s \left( \eps^{2H} \hat W_t+\hat h_t^y\right)
\dd [ \eps^{2H} \tilde W + \tilde h^y ]_t
-\frac{\eps^{1+2H}}{2} \int_0^1\sigma ^{2}\left( \eps^{2H}
\hat W_t+\hat h_t^y\right) \dd t,
\end{equation}
with stochastic Taylor expansion \eqref{taylor:expansion}.
From Girsanov theorem, $J_\delta(\eps,y) / \eps^{1-2H}$ equals
{  
\[
\begin{split}
&=
 e^{ \frac{1}{2\eps^{4H}} \| \h^y \|^2_{H^1 }}
\bE_\delta \bigg[\sigma(\eps^{2H}\hat{W}_1)^2 
\1_{\eps^{2H-1} X_1^\eps \geq y } \big( \int_0^1\frac{1}{\bar{\rho} \sigma(\eps^{2H} \hat{W}_t) } {\eps^{2H}}  
\dd \bar{W}_t\big) \bigg]
\\
&=
\bE_\delta \bigg[ e^{-\frac{1}{\eps^{2H}} \int_0^1 \dot{h}^y \dd W }
\sigma^2 (\eps^{2H} \hat{W}_1+\hat{h}_1^y) 
\1_{\eps^{2H} g_1  + \eps^{4H} g_2   + r_3 \geq 0}  \big( \int_0^1\frac{1}{\bar{\rho} \sigma(\eps^{2H} \hat{W}_t+\hat{h}_t^y) }(  {\eps^{2H}}  
\dd \bar{W}_t + \dd \bar{h}_t^y)\big)
\bigg].
\end{split}
\]

}
Theorem \ref{thm:remainder:estimate}, applied with $\eps^{2H}$ (instead of $\eps$), gives
on $\{\eps^{2H} |||\WW|||<\delta\}$,
\[
\sigma^2 (\eps^{2H} \hat{W}_1+\hat{h}_1^y)
=
\sigma^2 (\hat{h}_1^y)
+
\ell^1_{\eps,\WW},
\]
with $|\ell^1_{\eps, \WW}| \leq C \delta$ and
\[
\int_0^1\frac{1}{\bar{\rho} \sigma(\eps^{2H}  \hat{W}_t+\hat{h}_t^y) } ( { \eps^{2H}}    \dd \bar{W}_t + \dd \bar{h}_t^y ) = 
\int_0^1\frac{1}{\bar{\rho} \sigma(\hat{h}_t^y) } \dd \bar{h}_t^y +\ell^2_{\eps,\WW}
\]
with $|\ell^2_{\eps,\WW}| \leq C \eps^{2H}|||\WW||| \leq C \delta$. Therefore,
\[
\sigma^2 (\eps^{2H} \hat{W}_1+\hat{h}_1^y)
\int_0^1\frac{1}{\bar{\rho} \sigma(\eps^{2H}  \hat{W}_t+\hat{h}_t^y) } ( { \eps^{2H}}    \dd \bar{W}_t + \dd \bar{h}_t^y ) = 
\sigma^2 (\hat h_1^y)\int_0^1\frac{1}{\bar{\rho} \sigma(\hat{h}_t^y) } \dd \bar{h}_t^y +\ell_{\eps,\WW}
\]
with $|\ell_{\eps,\WW}|\leq C\delta$.
If $\eps^{2H}|||\WW|||\leq \delta$ we also have \eqref{eq:r}, so, for fixed $\delta$, for $\eps$ small enough, 
\[
|r_3^{\eps}| \leq \delta \eps^{4H} (C+|||\WW|||^2).
\]
 We have
\begin{equation}\label{eq:reduced}
\bE_\delta[\dots]
\in
\bigg(
\sigma^2 (\hat{h}_1^y)
\int_0^1\frac{\dd \bar{h}_t^y}{\bar{\rho} \sigma(\hat{h}_t^y) }  \pm C\delta\bigg)
\bE_\delta \bigg[ e^{-\frac{1}{\eps^{2H}} \int_0^1 \dot{h}^y \dd W }
\1_{ g_1  + \eps^{2H} g_2   \pm \delta \eps^{2H} (C+|||\WW|||^2) \geq 0} 
\bigg]
\end{equation}
The optimal condition \cite[Lemma C.3]{friz2018precise1} gives $\int_0^1 \dot{h}^y \dd W = \Lambda^{\prime}(y)g_1$.
By \cite[Lemma 8.3]{friz2018precise1},
\begin{equation}\label{dec:g2}
g_2 = \Delta_2 + g_1 \Delta_1 + g_1^2 \Delta_0
\end{equation}
where the $\Delta_i$'s are independent of $g_1$. We set now, as in \cite{azencott1985petites}, the zero mean Gaussian process $\V =\V^y$ 
\begin{equation} \label{defV}
       \V_t (\omega)  := \W_t (\omega) - g_1(\omega) \v_t
\end{equation}
where $\v$ is chosen so that $\V$ is independent of $g_1$. We also define
\[
\VV (\omega) := T_{- g_1(\omega) \v} \WW (\omega)
\]
where $T$, the ``lifted'' Cameron--Martin translation, is defined in \eqref{def:T}. As in Section 8.1 of \cite{friz2018precise1}, we let
\[\begin{split}
\tilde \Delta_{0}  &:= \Delta _{0} + C \delta \| \v \|^2_{H^1}, \\
\tilde \Delta^\pm_{2} &:= \Delta _{2} \pm \delta (C+ ||| \VV |||^2), 
\end{split}\]
where $\tilde \Delta _{2}^\pm$ is also $P$-independent of $g_1$  and $\VV$. (This independence allows for conditional Gaussian computations.) We refer to \cite{friz2018precise1} for details, and here we only use that
$\eps^{2H} ||| \VV ||| \leq C \delta,$
so that
$$|\eps^{2H} \Delta_1| \leq C \eps^{2H} ||| \VV |||
 \leq C \delta, $$
when $\eps^{2H} |||\WW|||\leq \delta$. Thus, the asymptotic behavior of $J_\delta (\eps, y)$ is sandwiched by $
\sigma^2 (\hat{h}_1^y)
\int_0^1\frac{\dd \bar{h}_t^y}{\bar{\rho} \sigma(\hat{h}_t^y) }  \pm C\delta$ times $(*)$ with 
\begin{eqnarray*}
      (*)   & \in & \bE_\delta \left[ \exp \left( {  - \frac{\Lambda^{\prime}(y)g_1 }{\eps^{2H}}  } \right) \1_{  g_1+ \eps^{2H}  \tilde \Delta^\pm_{2} \pm C (1+ \tilde \Delta _{0} ) \delta |g_1| >0}
\right] .
\end{eqnarray*}
The limit of this expectation can be computed with the Laplace method. 
We prove the upper bound. Clearly,
$$
      \bE_\delta \left[ \exp \left( {  - \frac{\Lambda^{\prime}(y)g_1 }{\eps^{2H}}  } \right)
      \1_{  g_1+ \eps^{2H}  \tilde \Delta^\pm_{2} \pm C (1+ \tilde \Delta _{0} ) \delta |g_1| >0}
\right] \le \bE \big[ \cdots \big] 
$$
where $\cdots$ means the same argument. Set $\sigma_y=\sqrt{2\Lambda(y)}/\Lambda^{\prime}(y)$ and
$$
\gamma_\delta:= C(1+\tilde \Delta_0) \delta
$$
and assume that $\delta$ is small enough that $\gamma_\delta < 1$. By \cite[Theorem 6.1, part (iii)]{friz2018precise1}, 
we have $\frac{\eps^{2H}}{\Lambda^{\prime}(y) \sigma_y}>0$ and can then apply Lemma \ref{lem:BSabs} (with $N= g_1 / \sigma_y$) to see that 
$$
\bE \big[ \cdots  |\Delta_2, \mathbf{V} \big]  \leq 
\frac{\eps^{2H}}{\sigma_y \Lambda^{\prime}(y)\sqrt{2\pi} } \max\left[  e^{\frac{\Lambda^{\prime }\left( y\right)\left(\Delta_{2}+ \delta \left(C+||| \VV |||^2\right)\right)}{1-\gamma_{\delta}}}  , e^{\frac{\Lambda^{\prime }\left( y\right)\left(\Delta_{2}+ \delta \left(C+||| \VV |||^2\right)\right)}{1+\gamma_{\delta}}}\right].
$$
By \cite[Proposition 8.6 and proof of Corollary 7.1]{friz2018precise1}, $\exp\left(\Lambda^{\prime}\left( y\right)\Delta_{2}\right) \in L^{1+}$ and by \cite[Lemma 8.3 (iv)]{friz2018precise1} $\exp(||| \VV |||^2) \in L^{0+}$, so that by letting successively $\eps$ and $\delta$ go to $0$ we obtain that
\begin{equation*}
\limsup_{\varepsilon \to 0} \eps^{-2H} \bE \big[ \cdots \big]  \leq 
{ \frac{1}{\sigma_y \Lambda^{\prime}(y)\sqrt{2\pi} } }
\bE \left[ \exp\left(\Lambda^{\prime}\left( y\right)\Delta_{2}\right) \right].
 \end{equation*} 
Recalling now 
$\sqrt{2 \Lambda(y)}=\Lambda^{\prime}(y) \sigma_y$,
and the prefactor $\sigma^2 (\hat{h}_1^y) \int_0^1\frac{\dd \bar{h}_t^y}{\bar{\rho} \sigma(\hat{h}_t^y) }  \pm C\delta$ in \eqref{eq:reduced} we have the upper bound. 
The lower bound is proved in the same way using the lower bound in Lemma \ref{lem:BSabs}. 
\end{proof}

\begin{lemma}\label{prop:algebraic:equivalence}
Fix $\delta > 0$. Then there exists $c=c_{y,\delta}>0$ such that
\[
\begin{split}
     | J_\delta (\eps, y) - J (\eps, y) | = \cO(\exp (- c / \bar \eps^2)).
\end{split}
\]
\end{lemma}

\begin{proof}
To this end, recall the sub-probability \eqref{def:P:delta} and introduce 
\[
B=\{  \eps^{2H} ||| \WW ||| \geq \delta     \}^c.
\]
We have
\begin{equation}
J (\eps, y) -J_\delta (\eps, y)
= \exp \left( {  \frac{\Lambda(y)}{\eps^{4H}}  } \right) 
\bE\big[\sigma(\eps^{2H}\hat{W}_1)^2 \1_{X_1^\eps \geq y \eps^{1-2H}} \big( \int_0^1\frac{1}{\bar{\rho} 
	\sigma(\eps^{2H} W^H_t) } \eps  \dd \bar{W}_t\big)\1_B\big].
\end{equation}
We have, for any $p,p'>1$ conjugate exponents,
\begin{equation}\label{eq:boundedsigma}
\begin{split}
&\bE\big|\sigma(\eps^{2H}\hat{W}_1)^2 \1_{X_1^\eps \geq y \eps^{1-2H}} \big( \int_0^1\frac{1}{\bar{\rho} \sigma(\eps^{2H} \hat{W}_t) } \eps  \dd \bar{W}_t\big)\1_B\big|
\\
&\leq \eps
\left[\bE\big|\sigma(\eps^{2H}\hat{W}_1)^2  \int_0^1\frac{1}{\bar{\rho} \sigma(\eps^{2H}  \hat{W}_t) }  \dd \bar{W}_t \big|^p\right]^{1/p}
\bE[ \1_{X_1^\eps \geq y \eps^{1-2H}} \1_B]^{1/p'}.
\end{split}
\end{equation}
The first factor can be bounded using H\"{o}lder inequality as
\[
\bE\big|\sigma(\eps^{2H}\hat{W}_1)^2  \int_0^1\frac{1}{\bar{\rho} \sigma(\eps^{2H}  \hat{W}_t) }  \dd \bar{W}_t \big|^p
\leq
(\bE\big[ \sigma(\eps^{2H}\hat{W}_1)^{2q}\big])^{1/q}
\Big(\bE\big|\int_0^1\frac{1}{\bar{\rho} \sigma(\eps^{2H}  \hat{W}_t) }  \dd \bar{W}_t \big|^{pq'}\Big)^{1/{q'}}.
\]
Since $\sigma(\cdot)$ satisfies \eqref{eq:C2}, the first factor is bounded for any $q>1$. 
{Using Burkholder--Davis--Gundy inequality, }
\[
\bE\big|\int_0^1\frac{1}{\sigma(\eps^{2H}  \hat{W}_t) }  \dd \bar{W}_t \big|^{pq'}
\leq
\bE\big|\int_0^1\frac{\dd t}{ \sigma(\eps^{2H}  \hat{W}_t)^2 }  \big|^{pq'/2}
\]
using condition \eqref{eq:C1} and the moment formula for log-normal variables,
\[
\dots \leq \bE\big|\int_0^1 \exp(c \eps^{2H}  \hat{W}_t) \dd t   \big|^{pq'/2}
\leq \bE\big| \exp(c \eps^{2H}  \hat{W}_1) \big|^{pq'/2} <\infty
\]
We conclude that the first factor in \eqref{eq:boundedsigma} is bounded by a constant, for any $p\geq1$. In \cite[lines after (8.5)]{friz2018precise1} it is shown that 
\[
\bE[ \1_{X_1^\eps \geq y \eps^{1-2H}} \1_B ]
=
\cO(e^{- \frac{\Lambda(\delta,y)}{\eps^{4H}}})
\]
with $\Lambda(\delta,y)>\Lambda(y)$. Now,
\[
\bE[ \1_{X_1^\eps \geq y \eps^{1-2H}} \1_B]^{1/p'}
=
\cO(e^{- \frac{\Lambda(\delta,y)}{p^{\prime} \eps^{4H}}})
\]
and we can choose $p'>1$ close enough to $1$ to have $\Lambda(\delta,y)/p'>\Lambda(y)$. The statement follows.
\end{proof}

\begin{lemma} \label{lem:BSabs}
Let  $\alpha \in \R$, $\gamma \in [0,1)$, $\varepsilon>0$, and $N \sim \cN(0,1)$. 
Then for some $C>0$, it holds that
\begin{align} \label{equ:BSabs}
  & \min\left[ e^{\frac{\alpha}{1-\gamma}},e^{\frac{\alpha}{1+\gamma}}\right] - \varepsilon^2 
\max\left[ e^{\frac{\alpha}{1-\gamma}}
\frac{\alpha^2-2\alpha(1-\gamma)+2(1-\gamma)^2}{(1-\gamma)^2}
,e^{\frac{\alpha}{1+\gamma}}
\frac{\alpha^2-2\alpha(1+\gamma)+2(1+\gamma)^2}{(1+\gamma)^2}
\right] \\
\label{equ:BSabs2}
  \leq  &\sqrt{2\pi} \varepsilon^{-1} \mathbb{E }\left[ \exp\left( -\varepsilon^{-1} N\right) \1_{ N + \gamma |N| + \varepsilon \alpha >0 }\right]  \\
  \label{equ:BSabs3}
 \leq  &   \max\left[ e^{\frac{\alpha}{1-\gamma}},e^{\frac{\alpha}{1+\gamma}}\right]
\end{align}
\end{lemma}

\begin{proof}
The middle expression (\ref{equ:BSabs2}) equals
$$
     \varepsilon^{-1}  \int_{-\infty}^{+\infty}  e^{-y^2/2} e^{-\frac{y}{\varepsilon}} \1_{y + \gamma |y| + \varepsilon \alpha>0} \dd y 
=   \int_{-\infty}^{+\infty}   e^{-v^2 \eps^2 /2}  e^{-v} \1_{v + \gamma | v| +  \alpha>0} \dd v.
$$
Inequalities $1 - y^2/2 \le \exp(-y^2/2) \le 1$ lead to the stated bounds. Indeed,
$$
           \text{(\ref{equ:BSabs2})} \leq   \int_{-\infty}^{+\infty}   e^{-v} \1_{v + \gamma | v| +  \alpha>0} \dd v
$$
which is computable. The right-hand side equals
\[
\begin{split}
\int_{-\frac{\alpha}{1+\gamma}}^{\infty} e^{-v} dv = e^{\frac{\alpha}{1 +\gamma}}, 
\mbox{
when } \alpha < 0 \\
\int_{-\frac{\alpha}{1-\gamma}}^{\infty}   e^{-v} \dd v 
= e^{\frac{\alpha}{1-\gamma}},
\mbox{
when }\alpha \geq 0.
\end{split}
\]
To obtain the lower bound, use $e^{-y^2/2} \geq 1 - y^2/2$ and split the integral to obtain
\[
 \text{(\ref{equ:BSabs2})} 
\geq   \int_{-\infty}^{+\infty}   e^{-v} \1_{v + \gamma | v| +  \alpha>0} \dd v  
- \frac{\varepsilon^{2}}{2} \int_{-\infty}^\infty e^{-v} v^2 \1_{v + \gamma | v| +  \alpha>0} \dd v.
\] 
The first integral is computed as before. For the second one, when $\alpha < 0$ we have
\[
e^{\frac{\alpha}{1+\gamma}}
\frac{\alpha^2-2\alpha(1+\gamma)+2(1+\gamma)^2}{(1+\gamma)^2}
\]
while  when $\alpha \geq 0$,
\[
e^{\frac{\alpha}{1-\gamma}}
\frac{\alpha^2-2\alpha(1-\gamma)+2(1-\gamma)^2}{(1-\gamma)^2}.
\]
The statement follows.
\end{proof}


\appendix

\section{Elements of regularity structures for rough volatility}\label{appendix:reg:str}

\noindent This appendix is based on \cite{bayer2017regularity,friz2018precise1}. 
We have an fBm $\hat W= K^H * \dot{W} $ of Hurst parameter $H$.
Let $M$ be the smallest integer such that
$ (M+1)H-1/2>0$ and then pick $\kappa$ small enough such that 
\begin{align}  \label{eq:condonkappa}
(M+1)(H-\kappa)-1/2-\kappa >0\,.
\end{align}
When $H=1/2$, we have $M=1$ and so $1/2-\kappa \in (1/3,1/2)$. This corresponds to 
the rough path case. 
More generally, we work with an enhancement of the Brownian noise $(W, \bar W)$, also known as a model of the form
\begin{equation}\label{def:model}
\WW (\omega) =  \left(W, \bar W, \hat W, \int \hat W \dd W , \int \hat W \dd \bar W,  \int \hat W^2 \dd W, \cdots, \int \hat W^M  \dd \bar W \right),
\end{equation}
with \emph{homogeneous model norm}\footnote{In fact, $\| W \|_{1/2-\kappa} \asymp   \| \hat W \|_{H-\kappa}$ by Schauder so that including $\hat W$ is mildly redundant.}$$
        ||| \WW |||  :=    \| W \|_{1/2-\kappa} + \| \bar W \|_{1/2-\kappa} +  \| \hat W \|_{H-\kappa} + \dots + \|  \int \hat W^M  \dd \bar W \|^{1/3}_{M(H-\kappa)+1/2 - \kappa}
$$
where  $\| \cdot|_{1/2-\kappa}$ are classical, resp. $2$-parameter, H\"older (semi)norms. 
One naturally defines, with $\h = (h,\bh)\in H^1$ and $\hat h=K^H*h$
\begin{equation}\label{def:T}
         T_{\h}  (\WW) =  \left(W+h, \bar W+ \bh, \hat W+ \hat h, \int (\hat W +  \hat h) d(W+\bh), ... \right) \ .
\end{equation}
Also, recall from \cite{bayer2017regularity} that there is a well-defined dilation $\delta_\eps$ acting on models. Formally, it is obtained by replacing 
each occurrence of $W, \bar W, \hat W$ with $\eps$ times that quantity:
 $$
          \delta_\eps \WW = \left(\eps W, \eps \bar{W},\eps \hat{W}, \eps^2 \int \hat W \dd W,  \eps^3 \int \hat W^2 \dd W , ....\right) \in \M ,
 $$
 where $\M$ is the space of models.
As a consequence, dilation works well with homogeneous model norms,
$$
 ||| \delta_\eps \WW ||| = \eps ||| \WW ||| \ .
$$

\begin{theorem}[Stochastic Taylor-like expansion]\label{thm:remainder:estimate}
Let $f$ be a smooth function. Fix $\h \in H^1$ and $ \eps >0$. If $\WW$ is a model, then so is $T_{\h} ( \delta_{\eps} \WW)$. 
The path-wise ``rough/model'' integral 
$$
\Psi(\eps) :=  \int_{0}^1 f \left( \eps \hat W_t + \hat{h}_t \right) 
\dd (T_{\h} ( \delta_{\eps} \WW))_t \ 
$$
is well-defined, continuously differentiable in $\eps$, and we have the estimates
\[
\begin{split}
|f(\eps\hat W_1 + \hat h_1) -
f(\hat h_1)|&=\mathcal{O}(\eps |||\WW||| ), \\
| \Psi(\eps) - \Psi(0) | &= \mathcal{O} (\eps ||| \WW ||| ),
\end{split}
\] 
valid on bounded sets of $\eps ||| \WW |||$.
\end{theorem}
\begin{proof}
As in \cite[Theorem B.6]{friz2018precise1}, just stop the expansion at the first order.
\end{proof}

\begin{lemma}
Let $\bar{Z}_1^\eps$ be defined in \eqref{def:Z} and recall $\hat \eps \equiv \eps^{2H}$. Then
\begin{equation}\label{taylor:expansion}
\bar{Z}_1^\eps=  g_0 + \hat \eps g_1 (\omega) + \hat \eps^2 g_2 (\omega) + r_3(\omega)
\end{equation}
with $g_0=y$,
\begin{align}
 \label{def:g1} 
g_1 
&= \int_0^1 \sigma^{\prime}(\hat{h}_s^y) \hat{W}_s \dd \tilde{h}_s^y 
+ \int_0^1 \sigma(\hat{h}_s^y) \dd \tilde{W}_s,
\\
\label{def:g2}
g_2 
&=
 \frac{1}{2}\int_0^1 \sigma^{\prime\prime}(\hat{h}_s^y) \hat{W}^2_s \dd \tilde{h}_s^y + \int_0^1\sigma^{\prime}(\hat{h}_s^y) \hat{W}_s \dd \tilde{W}_s
\\
\label{eq:r}
 | r_3 (\omega ) | & \le O (\eps^{6H}  ||| \WW |||^3) + O (\eps^{1+2H}),
\quad\mbox{uniformly on bounded sets of $\eps^{2H} ||| \WW |||$.}\end{align}
\end{lemma}
\begin{proof}
Directly from \eqref{def:Z},
$$
\bar{Z}_1^\eps=\int_0^1\s \left( \eps^{2H} \hat W_t+\hat h_t^y\right)
\dd [ \eps^{2H} \tilde W + \tilde h^y ]_t
+ O (\eps^{1+2H})
$$
uniformly on bounded sets of $\eps^{2H} |\tilde W|$, and hence on bounded sets of $\eps^{2H} ||| \WW |||$.
From \cite[Theorem B.6]{friz2018precise1}, applied with $\eps$ replaced by $\eps^{2H}$, and then again uniformly on bounded sets of $\eps^{2H} ||| \WW |||$
we arrive at the error estimate, 
$$
              | r_3 (\omega ) |  \le O (\eps^{6H}  ||| \WW |||^3) + O (\eps^{1+2H}),
$$
valid uniformly on bounded sets of $\eps^{2H} ||| \WW |||$.
\end{proof}

\end{document}